\documentclass[final]{IEEEtran}
\usepackage[latin9]{inputenc}
\usepackage{amsmath}
\usepackage{amssymb}
\usepackage{graphicx}
\usepackage{esint}

\makeatletter
\usepackage{amsfonts}
\usepackage{algorithm}
\usepackage{amsthm}
\providecommand{\U}[1]{\protect\rule{.1in}{.1in}}
\newtheorem{theorem}{Theorem}

\newtheorem{corollary}{Corollary}

\newtheorem{lemma}{Lemma}

\newtheorem{proposition}{Proposition}

\makeatother

\begin{document}

\title{Spatial Compressive Sensing for MIMO Radar}

\author{Marco Rossi,~\IEEEmembership{Student~Member,~IEEE,} Alexander M.
Haimovich,~\IEEEmembership{Fellow,~IEEE,} \and and Yonina C. Eldar,~\IEEEmembership{Fellow,~IEEE}%
\thanks{Copyright (c) 2013 IEEE. Personal use of this material is permitted.
However, permission to use this material for any other purposes must
be obtained from the IEEE by sending a request to pubs-permissions@ieee.org.

M. Rossi and A. H. Haimovich are with New Jersey Institute of Technology,
Newark, NJ 07102, USA. (e-mail:\ marco.rossi@njit.edu, haimovic@njit.edu).
The work of M. Rossi and A. M. Haimovich is partially supported by
the U.S. Air Force Office of Scientific Research under agreement No.
FA9550-12-1-0409.

Y. C. Eldar is with Technion - Israel Institute of Technology, Haifa
32000, Israel (e-mail: yonina@ee.technion.ac.il). The work of Y. Eldar
is supported in part by the Israel Science Foundation under Grant
no. 170/10, and in part by the Ollendorf Foundation.%
}}
\maketitle
\begin{abstract}
We study compressive sensing in the spatial domain to achieve target
localization, specifically direction of arrival (DOA), using multiple-input
multiple-output (MIMO) radar. A sparse localization framework is proposed
for a MIMO array in which transmit and receive elements are placed
at random. This allows for a dramatic reduction in the number of elements
needed, while still attaining performance comparable to that of a
filled (Nyquist) array. By leveraging properties of structured random
matrices, we develop a bound on the coherence of the resulting measurement
matrix, and obtain conditions under which the measurement matrix satisfies
the so-called \textit{isotropy} property. The coherence and isotropy
concepts are used to establish uniform and non-uniform recovery guarantees
within the proposed spatial compressive sensing framework. In particular,
we show that non-uniform recovery is guaranteed if the product of
the number of transmit and receive elements, $MN$ (which is also
the number of degrees of freedom), scales with $K\left(\log G\right)^{2}$,\ where
$K$ is the number of targets and $G$ is proportional to the array
aperture and determines the angle resolution. In contrast with a filled
virtual MIMO\ array where the product $MN$ scales linearly with
$G$, the logarithmic dependence on $G$ in the proposed framework
supports the high-resolution provided by the virtual array aperture
while using a small number of MIMO radar elements. In the numerical
results we show that, in the proposed framework, compressive sensing
recovery algorithms are capable of better performance than classical
methods, such as beamforming and MUSIC.
\end{abstract}
\begin{keywords} Compressive sensing, MIMO radar, random arrays,
direction of arrival estimation. \end{keywords}

\section{Introduction}

\PARstart{D}{etection}, localization, and tracking of targets
are basic radar functions. Limited data support and low signal-to-noise
ratios (SNR) are among the many challenges frequently faced by localization
systems. Another challenge is the presence of nearby targets, whether
in terms of location or Doppler, since closely spaced targets are
more difficult to discriminate. In multiple-input multiple-output
(MIMO) radar, targets are probed with multiple, simultaneous waveforms.
Relying on the orthogonality of the waveforms, returns from the targets
are jointly processed by multiple receive antennas. MIMO radar is
typically used in two antenna configurations, namely distributed \cite{alex}
and colocated \cite{stoica}. Depending on the mode of operation and
system architecture, MIMO radars have been shown to boost target detection,
enhance spatial resolution, and improve interference suppression.
These advantages are achieved by providing and exploiting a larger
number of degrees of freedom than \textquotedblleft{}conventional\textquotedblright{}
radar.

In this work, we focus on the application of colocated MIMO radar
to direction-of-arrival (DOA) estimation. It is well known in array
signal processing \cite{van trees} that DOA resolution improves by
increasing the array aperture. However, increasing the aperture without
increasing the number of sensors may lead to ambiguities, i.e., measurements
explained by erroneous sets of locations. A non-ambiguous uniform
linear array (ULA) must have its elements spaced at intervals no larger
than $\lambda/2$, where $\lambda$ is the signal wavelength. For
MIMO radar, unambiguous direction finding of targets is possible if
$N$ receive elements are spaced $\lambda/2$, and $M$ transmit elements
are spaced $N\lambda/2$, a configuration known as \textit{virtual
ULA} \cite{stoica}. In sampling parlance, the $\lambda/2$-spaced
array and the MIMO virtual ULA perform spatial sampling at the Nyquist
rate. The main disadvantage of this Nyquist setup is that the product
of the number of transmit and receive elements, $MN$, needs to scale
linearly with the array aperture, and thus with resolution.

In this paper, we propose the use of a sparse, random array architecture
in which a low number of transmit/receive elements are placed at random
over a large aperture. This setup is an example of spatial compressive
sensing since spatial sampling is applied at sub-Nyquist rates. The
goal of spatial compressive sensing is to achieve similar resolution
as a filled array, but with significantly fewer elements.

Localizing targets from undersampled array data links random arrays
to \textit{compressive sensing} \cite{EldarCS}. Random array theory
can be traced back to the 1960's. In \cite{Lo64}, it is shown that
as the number of sensors is increased, the random array pattern, a
well known quantity to radar practitioners, converges to its average.
This is because the array pattern's variance decreases linearly with
the number of elements. This work was extended to MIMO radar in \cite{Haleem08}.
The main conclusion of the classical random array literature was that
the random array pattern can be controlled by using a sufficient number
of sensors. However, two fundamental questions were left pending:
How many sensors are needed for localization as a function of the
number of targets, and which method should be used for localization?
Here we suggest that the theory and algorithms of compressed sensing
may be used to address these questions.

Early works on compressive sensing radar emphasize that the sparse
nature of many radar problems supports the reduction of temporal as
well as spatial sampling (an overview is given in \cite{Ender10}).
Recent work on compressive sensing for single-input single-output
radar \cite{EldarTime10}-\cite{EldarPrototype} demonstrates either
an increased resolution or a reduction in the temporal sampling rate.
Compressive sensing for MIMO radar has been applied both on distributed
\cite{Nehorai} and colocated \cite{Strohmer} setups. Much of the
previous literature on compressive sensing for colocated arrays discusses
the ULA\ setup, either within a passive system (with only receive
elements) \cite{Maliutov05} or in a MIMO radar \cite{Strohmer,Strohmer12}
setup. In particular, \cite{Strohmer12} imposes a MIMO radar virtual
ULA and derives bounds on the number of elements to perform range-angle
and range-Doppler-angle recovery by using compressive sensing techniques.
As discussed above, the (virtual) ULA setup performs Nyquist sampling
in the spatial domain. In contrast, we are interested in \emph{spatial}
compressive sensing (i.e., reducing the number of antenna elements
while fixing the array aperture), and rely on a random array geometry.
Links between compressive sensing and random arrays have been explored
in \cite{Carin09}. The author shows that spatial compressive sensing
can be applied to the passive DOA problem, allowing for a reduction
in the number of receiving elements. However, the MIMO radar framework
poses a major challenge: contrary to the passive setup, where the
rows of the sensing matrix $\mathbf{A}$ are independent, the MIMO
radar $MN$ measurements are dependent (they conform to the structure
of the MIMO random array steering vector). This lack of independence
prevents the application of the vast majority of results in the compressive
sensing literature. A MIMO radar random array architecture is studied
in \cite{Yu10}, but no recovery guarantees are provided.

Low-rate spatial sampling translates into cost savings due to fewer
antenna elements involved. It is of practical interest to determine
the least amount of elements required to guarantee correct targets
recovery. Finding conditions that guarantee recovery has been a main
topic of research, and it is one of the underpinnings of compressive
sensing theory. Recent work has shown that, for a sufficient number
of independent and identically distributed (i.i.d.) compressive sensing
measurements, non-uniform recovery can be guaranteed if a specific
property of the random sensing matrix, called \textit{isotropy}, holds
\cite{CandesRIPless}. While this property plays an important role,
this result does not apply to our setup since the MIMO radar $MN$
measurements are not independent. The dependent measurements problem
was recently addressed in \cite{RahutStromer}. There, the authors
derived conditions for non-uniform recovery using spatial compressive
sensing in a MIMO radar system with $N$ transceivers.

This work expands the literature in several ways. We propose a sparse
localization framework for a MIMO random array assuming a general
setup of $M$ transmitters and $N$ receivers. We provide a bound
on the coherence of the measurement matrix, and determine the conditions
under which the isotropy property holds. This allows us to develop
both uniform and non-uniform recovery guarantees for target localization
in MIMO radar systems. The proposed MIMO random array framework is
of practical interest to airborne and other radar applications, where
the spacing between antenna elements may vary as a function of aspect
angle towards the target, or where exact surveying of element locations
is not practical due to natural flexing of the structures involved.
Our results show that one can obtain the high-resolution provided
by a virtual array aperture while using a reduced number of antenna
elements.

The paper is organized as follows: Section \ref{Sec_Problem} introduces
the system model and the proposed sparse localization framework. Section
\ref{Sec_spatialCS} discusses spatial compressive sensing. Recovery
guarantees are derived in Section \ref{Sec_recovery}. In Section
\ref{Sec_Numerical}, we present numerical results demonstrating the
potential of the proposed framework, followed by conclusions in Section
\ref{s:conc}.

The following notation is used: boldface denotes matrices (uppercase)
and vectors (lowercase); for a vector $\mathbf{a}$, the $i$-th index
is $\mathbf{a}_{i}$, while for a matrix $\mathbf{A}$, the $i$-th
row is denoted by $\mathbf{A}\left(i,:\right)$. The complex conjugate
operator is $\left(\cdot\right)^{\ast}$, the transpose operator is
$\left(\cdot\right)^{T}$, and the complex conjugate-transpose operator
is $\left(\cdot\right)^{H}$. We define $\left\Vert \mathbf{X}\right\Vert _{0}$
as the number of non-zero norm rows of $\mathbf{X}$, the support
of $\mathbf{X}$ collects the indices of such rows, and a $K$-sparse
matrix satisfies\ $\left\Vert \mathbf{X}\right\Vert _{0}\leq K$.
The operator $\mathbb{E}$ denotes expectation and we define $\psi_{x}\left(u\right)\triangleq\mathbb{E}\left[\exp\left(jxu\right)\right]$
as the characteristic function of the random variable $x$. The symbol
``$\otimes$\textquotedblright{}\ denotes the Kronecker product.
The notation $\mathbf{x}\sim\mathcal{CN}\left(\mathbf{\mu},\mathbf{C}\right)$
means that the vector $\mathbf{x}$ has a circular symmetric complex
normal distribution with mean $\mathbf{\mu}$ and covariance matrix
$\mathbf{C}$. We denote by $K_{\alpha}\left(\cdot\right)$ the modified
Bessel function of the second kind.

\section{System Model\label{Sec_Problem}}

\subsection{MIMO Radar Model}

We model a MIMO\ radar system (see Fig. \ref{fig: system model})
in which $N$ sensors collect a finite train of $P$ pulses sent by
$M$ transmitters and returned from $K$ stationary targets. We assume
that transmitters and receivers each form a (possibly overlapping)
linear array of total aperture $Z_{TX}$ and $Z_{RX}$, respectively.
The quantities $Z_{TX}$ and $Z_{RX}$ are normalized in wavelength
units. Defining $Z\triangleq Z_{TX}+Z_{RX}$, the $m$-th transmitter
is at position $Z\xi_{m}/2$ on the $x$-axis, while the $n$-th receiver
is at position $Z\zeta_{n}/2$. Here $\xi_{m}$ lies in the interval
$\left[-\frac{Z_{TX}}{Z},\frac{Z_{TX}}{Z}\right]$, and $\zeta_{n}$
is in $\left[-\frac{Z_{RX}}{Z},\frac{Z_{RX}}{Z}\right]$. This definition
ensures that when $Z_{TX}=Z_{RX}$, both $\xi_{m}$ and $\zeta_{n}$
are confined to the interval $\left[-\frac{1}{2},\frac{1}{2}\right]$,
simplifying the notation in the sequel.

Let $s_{m}\left(t\right)$ denote the continuous-time baseband signal
transmitted by the $m$-th transmit antenna and let $\theta$ denote
the location parameter(s) of a generic target, for example, its azimuth
angle. Assume\ that the propagation is nondispersive and that the
transmitted probing signals are narrowband (in the sense that the
envelope of the signal does not change appreciably across the antenna
array). Then the baseband signal at the target location, considering
the $p$-th transmitted pulse, can be described by (see, e.g., \cite{alex})
\begin{equation}
\sum_{m=1}^{M}\exp\left(j2\pi f_{0}\tau_{m}\left(\theta\right)\right)s_{m}\left(t-pT\right)\triangleq\mathbf{c}^{T}\left(\theta\right)\mathbf{s}\left(t-pT\right).
\end{equation}
Here $f_{0}$ is the carrier frequency of the radar, $\tau_{m}\left(\theta\right)$
is the time needed by the signal emitted by the $m$-th transmit antenna
to arrive at the target, $\mathbf{s}\left(t\right)\triangleq\left[s_{1}\left(t\right),\ldots,s_{M}\left(t\right)\right]^{T}$,
$T$ denotes the pulse repetition interval, and
\begin{equation}
\mathbf{c}\left(\theta\right)=\left[\exp\left(j2\pi f_{0}\tau_{1}\left(\theta\right)\right),\ldots,\exp\left(j2\pi f_{0}\tau_{M}\left(\theta\right)\right)\right]^{T}
\end{equation}
is the transmit steering vector. Assuming that the transmit array
is calibrated, $\mathbf{c}\left(\theta\right)$ is a known function
of $\theta$.

To develop an expression for the received signal $r_{n}\left(t\right)$
at the $n$-th receive antenna, let
\begin{equation}
\mathbf{b}\left(\theta\right)=\left[\exp\left(j2\pi f_{0}\tilde{\tau}_{1}\left(\theta\right)\right),\ldots,\exp\left(j2\pi f_{0}\tilde{\tau}_{N}\left(\theta\right)\right)\right]^{T}
\end{equation}
denote the receive steering vector. Here $\tilde{\tau}_{n}\left(\theta\right)$
is the time needed for the signal reflected by the target located
at $\theta$ to arrive at the $n$-th receive antenna. Define the
vector of received signals as $\mathbf{r}\left(t\right)\triangleq\left[r_{1}\left(t\right),\ldots,r_{N}\left(t\right)\right]^{T}$.
Under the simplifying assumption of point targets, the received data
vector is described by \cite{alex}
\begin{equation}
\mathbf{r}\left(t\right)=\sum_{k=1}^{K}\sum_{p=1}^{P-1}x_{k,p}\mathbf{\mathbf{b}}\left(\theta_{k}\right)\mathbf{c}^{T}\left(\theta_{k}\right)\mathbf{s}\left(t-pT\right)+\mathbf{e}\left(t\right)\label{eq: receivedY}
\end{equation}
where $K$ is the number of targets that reflect the signals back
to the radar receiver, $x_{k,p}$ is the complex amplitude proportional
to the radar cross sections of the $k$-th target relative to pulse
$p$-th, $\theta_{k}$ are locations, and $\mathbf{e}\left(t\right)$
denotes the interference plus-noise term. The targets' positions are
assumed constant over the observation interval of $P$ pulses. We
assume that the target gains $\left\{ x_{k,p}\right\} $ follow a
Swerling Case II model, meaning that they are fixed during the pulse
repetition interval $T$, and vary independently from pulse to pulse
\cite{Skolnik}.

Analyzing how to estimate the number of targets $K$, or the noise
level, without prior information is the topic of current work \cite{Rossi12b},
but outside the scope of this paper. Therefore in the following, we
assume that the number of targets $K$\ is known and the noise level
is available.%TCIMACRO{\FRAME{ftbpFU}{3.2093in}{1.7781in}{0pt}{\Qcb{MIMO radar system
%model.}}{\Qlb{fig: system model}}{system_model.eps}%
%width 3.2093in;  height 1.7781in;  depth 0pt;  original-width 6.5933in;
%original-height 3.627in;  cropleft "0";  croptop "1";  cropright "1";
%filename 'graphics/system_model.eps';file-properties "XNPEU";}}}%
%BeginExpansion
\begin{figure}
\centering{}\includegraphics[width=3.2093in,height=1.7781in]{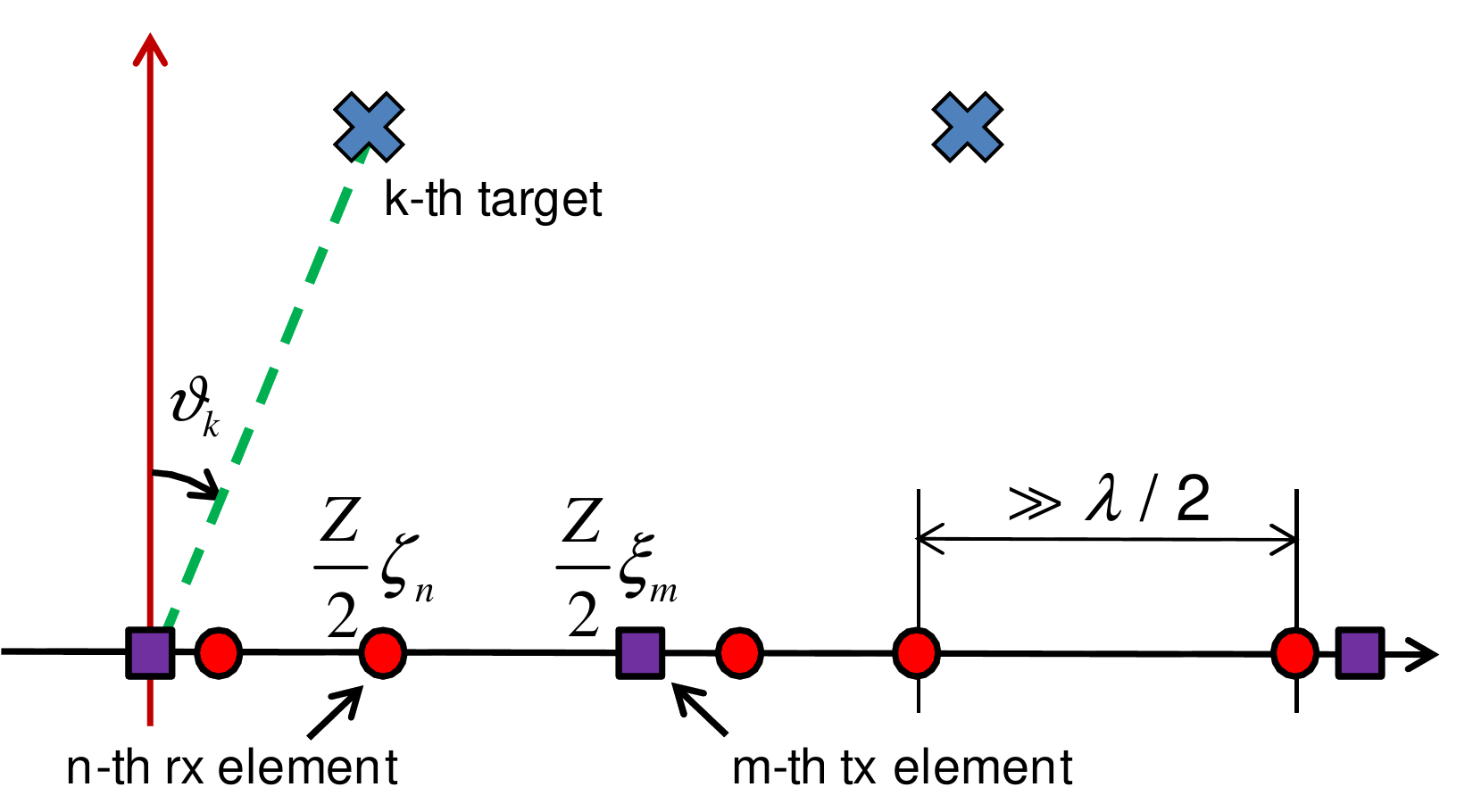}\caption{MIMO radar system model.}
\label{fig: system model}
\end{figure}

%EndExpansion

\subsection{Problem Formulation}

The purpose of the system is to determine the DOA angles to targets
of interest. We consider targets associated with a particular range
and Doppler bin. Targets in adjacent range-Doppler bins contribute
as interferences to the bin of interest. The assumption of a common
range bin implies that all waveforms are received with the same time
delay after transmission. Since range and Doppler measurements are
not of interest, the common time delay and Doppler shift are not explicitly
shown in our model. This approach is justified because angle resolution
is essentially independent of range-Doppler resolution in antenna
arrays \cite{DiFranco65}. Being capable to handle targets with non-zero
Doppler, our approach is applicable to airborne or ground targets.
Targets are assumed in the far-field, meaning that a target's DOA
parameter $\theta\triangleq\sin\vartheta$ (where $\vartheta$ is
the DOA angle) is constant across the array. Under these assumptions,
the receiver and transmitter steering vectors, $\mathbf{b}\left(\theta\right)$
and $\mathbf{c}\left(\theta\right)$ respectively, become
\begin{equation}
\mathbf{b}\left(\theta\right)=\left[\exp\left(j\pi Z\theta\zeta_{1}\right),\ldots,\exp\left(j\pi Z\theta\zeta_{N}\right)\right]^{T}\label{e:b}
\end{equation}
and
\begin{equation}
\mathbf{c}\left(\theta\right)=\left[\exp\left(j\pi Z\theta\xi_{1}\right),\ldots,\exp\left(j\pi Z\theta\xi_{M}\right)\right]^{T}.\label{e:c}
\end{equation}

By cross-correlating the received signal at each sensor with filters
matched to each of the probing waveforms, we obtain
\begin{align}
\mathbf{y}_{p} & =\operatorname{vec}\left[\int\mathbf{r}\left(t\right)\mathbf{s}^{H}\left(t-pT\right)dt\right]\label{eq: yp}\\
 & =\operatorname{vec}\left[\sum\nolimits _{k=1}^{K}\sum\nolimits _{p=0}^{P-1}x_{k,p}\mathbf{\mathbf{b}}\left(\theta_{k}\right)\mathbf{c}^{T}\left(\theta_{k}\right)\mathbf{W+}\right.\nonumber \\
 & \text{\hspace{1.2in}}\left.+\int\mathbf{e}\left(t\right)\mathbf{s}^{H}\left(t-pT\right)dt\right]\nonumber 
\end{align}
where the $M\times M$ matrix $\mathbf{W}$ has elements
\begin{equation}
\left[\mathbf{W}\right]_{m,j}=\mathbf{\int}s_{m}\left(t\right)s_{j}^{\ast}\left(t\right)dt.\label{e:W}
\end{equation}
We assume the $M$ probing waveforms to be orthogonal (e.g., pulses
modulated by an orthogonal code), therefore\ $\mathbf{W}=\mathbf{I}$.
Defining the $MN\times P$ matrix $\mathbf{Y}\triangleq\left[\mathbf{y}_{1},\ldots,\mathbf{y}_{P}\right]$,
we have from (\ref{eq: yp})
\begin{equation}
\mathbf{Y}=\mathbf{\tilde{A}}\left(\mathbf{\theta}\right)\mathbf{\tilde{X}}+\mathbf{E.}\label{e:Y}
\end{equation}
Here $\mathbf{\tilde{X}}=\left[\mathbf{\tilde{x}}_{1},\ldots,\mathbf{\tilde{x}}_{P}\right]$
is a $K\times P$ matrix with $\mathbf{\tilde{x}}_{p}=\left[x_{1,p},\ldots,x_{K,p}\right]^{T}$,
\begin{equation}
\mathbf{\tilde{A}}\left(\mathbf{\theta}\right)=\left[\mathbf{a}\left(\theta_{1}\right),\ldots,\mathbf{a}\left(\theta_{K}\right)\right]\label{eq: manifold}
\end{equation}
is a $MN\times K$ matrix with columns
\begin{equation}
\mathbf{\mathbf{a}}\left(\theta\right)\triangleq\mathbf{c}\left(\theta\right)\otimes\mathbf{\mathbf{b}}\left(\theta\right)\label{eq: ateta}
\end{equation}
known as the ``virtual array\textquotedblright{}\ steering vector,
and $\mathbf{E}=\left[\mathbf{e}_{1},\ldots,\mathbf{e}_{P}\right]$
is $MN\times P$ with $\mathbf{e}_{p}=\operatorname{vec}\left[\int\mathbf{e}\left(t\right)\mathbf{s}^{H}\left(t-pT\right)dt\right]$.
The term ``virtual array\textquotedblright{}\ indicates that $\mathbf{\mathbf{a}}\left(\theta\right)$\ can
be thought of as a steering vector with $MN$ elements.

Our aim is to recover $\mathbf{\theta}$ and $\mathbf{\tilde{X}}$\ from
$\mathbf{Y}$ using a small number of antenna elements. To do this,
we use a sparse localization framework. Neglecting the discretization
error, it is assumed that the target possible locations $\mathbf{\theta}$
comply with a grid of $G$ points $\phi_{1:G}$ (with $G\gg K$).
Since each element of $\mathbf{\theta}$ parameterizes one column
of $\mathbf{\tilde{A}}\left(\mathbf{\theta}\right)$, it is possible
to define an $MN\times G$ dictionary matrix $\mathbf{A}=\left[\mathbf{a}_{1},\ldots,\mathbf{a}_{G}\right]$,
where $\mathbf{a}_{g}=\mathbf{a}\left(\phi_{g}\right)$. From (\ref{eq: ateta}),
the steering vector $\mathbf{a}_{g}$ is the Kronecker product of
the receive steering vector $\mathbf{b}_{g}=\mathbf{b}\left(\phi_{g}\right)$
and the transmit steering vector $\mathbf{c}_{g}=\mathbf{c}\left(\phi_{g}\right)$:
\begin{equation}
\mathbf{a}_{g}=\mathbf{c}_{g}\otimes\mathbf{b}_{g}.\label{e:a}
\end{equation}
The received signal is then expressed as
\begin{equation}
\mathbf{Y}=\mathbf{AX}+\mathbf{E,}\label{eq: y}
\end{equation}
where the unknown $G\times P$ matrix $\mathbf{X}$ contains the target
locations and gains. Zero rows of $\mathbf{X}$ correspond to grid
points without a target. The system model (\ref{eq: y}) is sparse
in the sense that $\mathbf{X}$ has only $K\ll G$ non-zero rows.

Note that in the sparse localization framework, the matrix $\mathbf{A}$
is known, whereas in the array processing model (\ref{e:Y}), the
matrix $\mathbf{\tilde{A}}\left(\mathbf{\theta}\right)$ is unknown.
Given the measurements $\mathbf{Y}$ and matrix $\mathbf{A}$, our
goal translates into determining the non-zero norm rows' indices of
$\mathbf{X}$, i.e., the support of $\mathbf{X}$. The matrix $\mathbf{A}$
is governed by the choice of grid points $\phi_{1:G}$, by the number
$M$ of transmitters and their positions, $\xi_{1:M}$, and by the
number $N$ of receivers and their positions, $\zeta_{1:N}$. In the
following we assume that the transmitter (receiver) elements' positions
$\xi_{1:M}$ ($\zeta_{1:N}$) are independent and identically distributed
(i.i.d.) random variables governed by a probability density function
(pdf) $p\left(\xi\right)$ ($p\left(\zeta\right)$).

\section{Spatial Compressive Sensing Framework\label{Sec_spatialCS}}

The aim of spatial compressive sensing is to recover the unknown $\mathbf{X}$
from the measurements $\mathbf{Y}$ (see (\ref{eq: y})) using a small
number of antenna elements, $MN$, while fixing the array aperture
$Z$. In this section we introduce the proposed spatial compressive
sensing framework and overview practical recovery algorithms (the
well-known beamforming method as well as compressive sensing based
algorithms).

\subsection{Beamforming}

Consider the scenario in which the transmitters and receivers locations
support the Nyquist array (virtual ULA) geometry. In this setting,
the matrix $\mathbf{A}$ in (\ref{eq: y}) has a Vandermonde structure,
and the aperture scales linearly with the number of antenna elements,
$Z=\left(MN-1\right)/2$. If we choose the (uniform) grid of possible
target locations $\phi_{g}$ to match the array resolution, that is
$G=2Z+1$, then the matrix $\mathbf{\mathbf{A}}$ becomes a Fourier
matrix. In this case, $\mathbf{Q}\triangleq\mathbf{A}^{H}\mathbf{A}=MN\cdot\mathbf{I}$.
It follows that $\mathbf{X}$ can be estimated as $\left(1/MN\right)\cdot\mathbf{A}^{H}\mathbf{Y}$.
In array processing, this method is called beamforming. The support
of the unknown $\mathbf{X}$ is recovered by looking for peak values
of $\left\Vert \mathbf{a}_{g}^{H}\mathbf{Y}\right\Vert _{2}$ over
the grid points. Beamforming is also applied to estimate the locations
of targets not limited to a grid. This is done by finding the peaks
of $\left\Vert \mathbf{\mathbf{a}}^{H}\left(\theta\right)\mathbf{Y}\right\Vert _{2}$,
where $\mathbf{\mathbf{a}}\left(\theta\right)$ is a steering vector
(\ref{eq: ateta}) swept over the angles of interest. The shortcoming
of the Nyquist array setup is that the number of elements $MN$ must
scale linearly with the array aperture $Z$ and consequently, with
the resolution (i.e., such sampling mode requires $MN=G$).

Spatial compressive sensing implies that a sparse $\mathbf{X}$ can
be recovered from a number of spatial measurements significantly lower
than the Nyquist array, i.e., $MN\ll G$. The idea is to design the
sensing procedure so that the matrix $\mathbf{\mathbf{Q}}$ is a scalar
multiple of the identity matrix \emph{on average}%
\footnote{For instance, this is obtained when using a partial Fourier matrix.%
}, i.e., $\mathbb{E}\left[\mathbf{Q}\right]=MN\cdot\mathbf{I}$, and
to control the variance of the non-diagonal elements by using a sufficient
number of measurements. Intuitively, the more measurements $MN$ we
employ, the closer we get to a diagonal $\mathbf{\mathbf{Q}}$. Furthermore,
because when $MN<G$, each realization of $\mathbf{\mathbf{Q}}$ has
non-zero off-diagonal terms, the beamforming metric $\left\Vert \mathbf{a}_{g}^{H}\mathbf{Y}\right\Vert _{2}$
is affected not only by the $g$-th row of $\mathbf{X}$ and by the
noise, but also by any row of $\mathbf{X}$ that has non-zero norm.
This entails that, instead of beamforming, we resort to more sophisticated
recovery algorithms, which take advantage of the signal's sparsity
to mitigate the mutual interference among non-zero rows of $\mathbf{X}$.
A brief overview of compressive sensing recovery methods is provided
next.

\subsection{Compressive Sensing}

One way to classify compressive sensing models is according to the
number of pulses $P$ (``snapshots\textquotedblright{}\ in array
processing problems): single measurement vector (SMV) for $P=1$ ($\mathbf{Y}$
reduces to a single vector), or multiple measurement vector (MMV)
for $P\geq1$ ($\mathbf{Y}$ is a matrix). The system model in (\ref{eq: y})
is an example of an MMV setting. For simplicity, in the following
we consider an SMV scenario (i.e., $P=1$, $\mathbf{Y}=\mathbf{y}$,
$\mathbf{X}=\mathbf{x}$ and $\mathbf{E}=\mathbf{e}$ in (\ref{eq: y})).

In principle, a sparse $\mathbf{x}$ (i.e., it has only $K\ll G$
non-zero rows) can be recovered from the least number of elements
$MN$ by solving the non-convex combinatorial $\ell_{0}$-norm problem
\begin{equation}
\min_{\mathbf{x}}\left\Vert \mathbf{x}\right\Vert _{0}\text{ \ s.t. }\left\Vert \mathbf{y}-\mathbf{Ax}\right\Vert _{2}\leq\sigma,\label{eq: L0}
\end{equation}
or one of its equivalent formulations: a cardinality-constrained formulation,
$\min_{\mathbf{x}}\left\Vert \mathbf{y}-\mathbf{Ax}\right\Vert _{2}^{2}$
\ s.t. \ $\left\Vert \mathbf{x}\right\Vert _{0}\leq\gamma$, or
a Lagrangian formulation, $\min_{\mathbf{x}}\frac{1}{2}\left\Vert \mathbf{y}-\mathbf{Ax}\right\Vert _{2}^{2}+\nu\left\Vert \mathbf{x}\right\Vert _{0}$.
These three formulations are equivalent for a proper choice of the
parameters $\gamma$, $\nu$ and $\sigma$, which depend on prior
information, e.g. the noise level $\left\Vert \mathbf{e}\right\Vert _{2}$
or the sparsity $K$. Unfortunately, the solution to any of these
formulations requires an exhaustive search among all combinations
of non-zero norm indices of $\mathbf{x}$, necessitating exponential
complexity \cite{EldarCS}.

A variety of polynomial complexity algorithms have been proposed for
obtaining an approximate solution to (\ref{eq: L0}). One family of
methods is Matching Pursuit (MP). In its simplest version, an empty
provisional support is refined by adding one grid-point index at each
iteration. Among the matching pursuit algorithms, the most notable
in the SMV setting are Orthogonal Matching Pursuit (OMP) \cite{OMP},
Orthogonal Least Squares (OLS) \cite{bib: OLS}, and CoSaMP \cite{CoSaMP}.
For the general MMV setting, examples are the Rank Aware-Orthogonal
Recursive Matching Pursuit (RA-ORMP) algorithm \cite{Davies10} and
its generalization, Multi-Branch Matching Pursuit (MBMP) \cite{Rossi12}.
Another family of methods is known as Basis Pursuit (BP). The BP strategy
relaxes the $l_{0}$-norm in (\ref{eq: L0}) with the $l_{1}$-norm
\cite{Candes08}. The reformulation is known as LASSO, defined by
\begin{equation}
\min_{\mathbf{x}}\left\Vert \mathbf{x}\right\Vert _{1}\text{ \ s.t. }\left\Vert \mathbf{y}-\mathbf{Ax}\right\Vert _{2}\leq\sigma.\label{eq: lasso2}
\end{equation}
Unlike (\ref{eq: L0}), this problem is convex, and a global solution
can be found in polynomial time. Since it is a relaxation, the solution
obtained could be different from that of (\ref{eq: L0}). Finding
conditions that guarantee correct recovery with a specific method
(e.g. LASSO) has been a main topic of research and one of the underpinnings
of compressive sensing theory \cite{EldarCS}.

Two kinds of recovery guarantees are defined in compressive sensing:
uniform and non-uniform. A uniform recovery guarantee (addressed below
by Theorem \ref{Theorem MN coherence}) means that for a fixed instantiation
of the random measurement matrix $\mathbf{A}$, all possible $K$-sparse
signals are recovered with high probability. In contrast, a non-uniform
recovery result (addressed by Theorem \ref{theorem Rauhut}) captures
the typical recovery behavior for a random measurement matrix $\mathbf{A}$.
Specifically, suppose we are given an arbitrary $K$-sparse vector
$\mathbf{x}$, and we then draw $\mathbf{A}$ at random (independent
of $\mathbf{x}$). Non-uniform recovery details under what conditions
an algorithm will recover $\mathbf{x}$ with high probability. Note
that, for a non-uniform guarantee, $\mathbf{A}$ is being asked to
recover only a specific $\mathbf{x}$, not \emph{any} $K$-sparse
vectors. Therefore, uniform recovery implies non-uniform recovery,
but the converse is not true.

Loosely speaking, a uniform recovery guarantee can be obtained if,
with high probability, the matrix $\mathbf{A}$ has small coherence
\cite{EldarCS}. The coherence is defined as the maximum inner product
between the normalized columns of $\mathbf{A}$,
\begin{equation}
\mu\triangleq\max_{i\neq l}\frac{\left\vert \mathbf{a}_{i}^{H}\mathbf{a}_{l}\right\vert }{\left\Vert \mathbf{a}_{i}\right\Vert _{2}\left\Vert \mathbf{a}_{l}\right\Vert _{2}}.\label{eq: mu}
\end{equation}
Alternatively, uniform recovery is guaranteed if $\mathbf{A}$ satisfies
the Restricted Isometry Property (RIP) \cite{EldarCS} with high probability.
Non-uniform recovery follows if a specific property of the random
measurement matrix $\mathbf{A}$, called isotropy, holds \cite{CandesRIPless}.
The isotropy property states that the components of each row of $\mathbf{A}$\ have
unit variance and are uncorrelated, i.e.,
\begin{equation}
\mathbb{E}\left[\mathbf{A}^{H}\left(t,:\right)\mathbf{A}\left(t,:\right)\right]=\mathbf{I}\label{eq: isotropy}
\end{equation}
for every $t$.

Both (\ref{eq: mu}) and (\ref{eq: isotropy}) suggest that the matrix
$\mathbf{Q}\triangleq\mathbf{A}^{H}\mathbf{A}$ plays a key role in
establishing recovery guarantees. Indeed, because in our setting the
rows of $\mathbf{A}$ are identically distributed, a simple calculation
shows that $\mathbb{E}\left[\mathbf{A}^{H}\left(t,:\right)\mathbf{A}\left(t,:\right)\right]=\frac{1}{MN}\mathbb{E}\left[\mathbf{Q}\right]$,
thus the isotropy property requires $\mathbb{E}\left[\mathbf{Q}\right]=MN\cdot\mathbf{I}$.
Furthermore, as evident by the definition of coherence in (\ref{eq: mu}),
$\mu$ is the maximum absolute value among normalized off-diagonal
elements of $\mathbf{Q}$. In Section \ref{Sec_recovery}, by deriving
statistics of the matrix $\mathbf{Q}$, we provide conditions on design's
quantities ($p\left(\xi\right)$, $p\left(\zeta\right)$, $M$, $N$
and $\phi_{1:G}$) to obtain (uniform and non-uniform) recovery guarantees
for spatial compressive sensing.

\section{Recovery Guarantees\label{Sec_recovery}}

In this section, we develop recovery guarantees for sparse localization
with MIMO random arrays. In detail, we show how to choose the grid-points
$\phi_{1:G}$, the number of elements $MN$ and the distributions
governing the element positions $p\left(\xi\right)$ and $p\left(\zeta\right)$,
in order to guarantee\ target localization by spatial compressive
sensing via (\ref{eq: lasso2}). Due to the role of the matrix $\mathbf{Q}$
in recovery guarantees, we start by studying the statistics of $\mathbf{Q}$.

\subsection{Statistics of $\mathbf{Q}\triangleq\mathbf{A}^{H}\mathbf{A}$}

To study the statistics of $\mathbf{Q}$, we first analyze its relationship
to a quantity known to\ radar practitioners as the \emph{array pattern}
\cite{Johnson93}. In array processing, the array pattern $\beta\left(u_{i,l}\right)$
is the system response of an array beamformed in direction $\phi_{l}$
to a unit amplitude target located in direction $\phi_{i}$. In other
words, $\beta\left(u_{i,l}\right)$ is the inner product between two
normalized columns of the measurement matrix:
\begin{align}
\beta\left(u_{i,l}\right) & \triangleq\frac{\mathbf{a}_{i}^{H}\mathbf{a}_{l}}{\left\Vert \mathbf{a}_{i}\right\Vert _{2}\left\Vert \mathbf{a}_{l}\right\Vert _{2}}\label{eq: beta}\\
 & =\frac{1}{MN}\sum_{m=1}^{M}\sum_{n=1}^{N}\exp\left[ju_{i,l}\left(\zeta_{n}+\xi_{m}\right)\right],\nonumber 
\end{align}
where we defined
\begin{equation}
u_{i,l}\triangleq\pi Z\left(\phi_{l}-\phi_{i}\right).\label{eq: u}
\end{equation}
The peak of the absolute value of the array pattern for a target colinear
with the beamforming direction, $\left\vert \beta\left(0\right)\right\vert $,\emph{\ }is
called the \emph{mainlobe}.\emph{\ }Peaks of $\left\vert \beta\left(u\right)\right\vert $
for $u\neq0$, are known as \emph{sidelobes,} and the highest among
all the sidelobes is called the \emph{peak sidelobe}. Thus the terms
$\mathbf{a}_{i}^{H}\mathbf{a}_{l}$ in (\ref{eq: beta}) play the
role of sidelobes.

The relation between coherence, isotropy and array pattern is apparent.
Indeed, from (\ref{eq: mu}), (\ref{eq: beta}), and the definition
of sidelobes, the coherence, in array processing parlance, is the
peak sidelobe associated with the matrix $\mathbf{A}$. Similarly,
from (\ref{eq: isotropy}) and (\ref{eq: beta}), the isotropy can
be related to the mean array pattern 
\begin{equation}
\eta\left(u_{i,l}\right)\triangleq\mathbb{E}\left[\beta\left(u_{i,l}\right)\right],
\end{equation}
where the expectation $\mathbb{E}\left[\beta\left(u_{i,l}\right)\right]$
is taken with respect to the ensemble of element locations. In particular,
isotropy requires that $\eta\left(u_{i,l}\right)=0$ for any $i\neq l$.

For a system with randomly placed sensors, the array pattern $\beta\left(u_{i,l}\right)$
is a stochastic process. Naturally, statistics of the array pattern
of a random array depend on the pdf of the sensor locations. In \cite{Haleem08},
the authors derive the means and the variances of the real and imaginary
parts of $\beta\left(u_{i,l}\right)$. The following proposition formalizes
pertinent results from \cite{Haleem08}. For the sake of brevity,
we drop the dependency on $i$ and $l$, and denote the array pattern
as $\beta\left(u\right)$. Define\ $z\triangleq\zeta+\xi$, and assume
that the pdf of $z$,\ $p\left(z\right)$,\ is an even function
(so that $\operatorname{Im}\eta\left(u\right)=0$). Further, define
the variances of the array pattern $\sigma_{1}^{2}\left(u\right)\triangleq\operatorname*{var}\left[\operatorname{Re}\beta\left(u\right)\right]$,
$\sigma_{2}^{2}\left(u\right)\triangleq\operatorname*{var}\left[\operatorname{Im}\beta\left(u\right)\right]$
and $\sigma_{12}\left(u\right)\triangleq\mathbb{E}\left[\left(\operatorname{Re}\beta\left(u\right)-\eta\left(u\right)\right)\operatorname{Im}\beta\left(u\right)\right]$.

\begin{proposition} \label{Lemma beampattern}Let the locations $\xi$
of the transmit elements be i.i.d., drawn from a distribution $p\left(\xi\right)$,
and the locations $\zeta$ of the receive elements be i.i.d., drawn
from a distribution $p\left(\zeta\right)$. Then, for a given $u$,
the following holds:
\begin{enumerate}
\item The mean array pattern is the characteristic function of $z$, i.e.,
\begin{equation}
\eta\left(u\right)=\psi_{z}\left(u\right).\label{eqn:mean_pat_app}
\end{equation}

\item If $\xi$ and $\zeta$ are identically distributed, then (\ref{e:sig1}),
(\ref{e:sig2}), and $\sigma_{12}\left(u\right)=0$ hold. 
\end{enumerate}
\end{proposition}

\begin{proof} See Appendix A. \end{proof}

Proposition \ref{Lemma beampattern} links the probability distributions
$p\left(\xi\right)$ and $p\left(\zeta\right)$ (via $\psi_{z}\left(u\right)$
and $\psi\left(u\right)$) to the mean and variances of each element
of the matrix $\mathbf{Q}$, i.e., $\beta\left(u_{i,l}\right)=\frac{1}{MN}\mathbf{a}_{i}^{H}\mathbf{a}_{l}$.
As shown below, this result is used to obtain non-uniform recovery
guarantees. 
\begin{table*}[t]
\begin{align}
\sigma_{1}^{2}\left(u\right) & =\frac{1}{2MN}\left(1+\psi_{z}\left(2u\right)\right)+\psi_{z}\left(u\right)\left[\frac{N+M-2}{2MN}\left(1+\psi_{\xi}\left(2u\right)\right)-\psi_{z}\left(u\right)\frac{N+M-1}{MN}\right]\label{e:sig1}\\
\sigma_{2}^{2}\left(u\right) & =\frac{1}{2MN}\left(1-\psi_{z}\left(2u\right)\right)+\psi_{z}\left(u\right)\frac{N+M-2}{2MN}\left(1+\psi_{\xi}\left(2u\right)\right)\label{e:sig2}
\end{align}
 \hrulefill{}%The spacer can be tweaked to stop underfull vboxes.
\end{table*}

To characterize the statistics of the coherence $\mu$ (defined in
(\ref{eq: mu})), we need the distribution of the maximum absolute
value among normalized off-diagonal elements of $\mathbf{Q}$.

We now show that, by imposing specific constraints on the grid points
$\phi_{1:G}$ and on the probability distributions $p\left(\xi\right)$
and $p\left(\zeta\right)$, we can characterize the distributions
of the elements of $\mathbf{Q}$. To do this, we require an intermediate
result about the structure of the matrix $\mathbf{Q}$ when $\phi_{1:G}$
is a uniform grid:

\begin{lemma} \label{Lemma Toeplitz}If $\phi_{1:G}$ is a uniform
grid, $\mathbf{Q}$ is a Toeplitz matrix. \end{lemma}

\begin{proof} See Appendix B. \end{proof}

Thanks to Lemma \ref{Lemma Toeplitz}, whenever $\phi_{1:G}$ is a
uniform grid, $\mathbf{Q}$ is described completely by the elements
of the first row of $\mathbf{A}$, $\mathbf{a}_{1}^{H}\mathbf{a}_{i}$
for $i=1,\ldots,G$. From the definition of $\mathbf{A}$, its columns
all have squared-norm equal to $MN$. Therefore\ the elements on
the main diagonal of $\mathbf{Q}$ are equal to $MN$. Thus, we need
to investigate the remaining random elements, $\mathbf{a}_{1}^{H}\mathbf{a}_{i}$
for $i=2,\ldots,G$. By exploiting the Kronecker structure of the
columns of $\mathbf{A}$ in (\ref{e:a}), we can express the elements
of $\mathbf{Q}$ as:
\begin{align}
\mathbf{\mathbf{a}}_{i}^{H}\mathbf{\mathbf{a}}_{j} & \triangleq\left(\mathbf{c}_{i}^{H}\otimes\mathbf{b}_{i}^{H}\right)\left(\mathbf{c}_{j}\otimes\mathbf{b}_{j}\right)\nonumber \\
 & =\mathbf{\mathbf{c}}_{i}^{H}\mathbf{\mathbf{c}}_{j}\mathbf{\mathbf{b}}_{i}^{H}\mathbf{\mathbf{b}}_{j},\label{eq: aa bbcc}
\end{align}
where $\mathbf{b}$ and $\mathbf{c}$ are the steering vectors of
the receiver and transmitter arrays, respectively.

From (\ref{eq: aa bbcc}), the random variable $\beta\left(u_{1,i}\right)\triangleq\frac{1}{MN}\mathbf{a}_{1}^{H}\mathbf{a}_{i}$
is the product between the random variables $\beta_{\zeta}\left(u_{1,i}\right)\triangleq\frac{1}{N}\mathbf{b}_{1}^{H}\mathbf{b}_{i}$
and $\beta_{\xi}\left(u_{1,i}\right)\triangleq\frac{1}{M}\mathbf{c}_{1}^{H}\mathbf{c}_{i}$.
As such, the distribution of $\beta\left(u_{1,i}\right)$ (or equivalently,
$\mathbf{a}_{1}^{H}\mathbf{a}_{i}$) can be characterized from the
distributions of $\beta_{\zeta}\left(u_{1,i}\right)$ and $\beta_{\xi}\left(u_{1,i}\right)$.
Following the approach in \cite{Lo64}, we show in Appendix C that
the real and imaginary parts of $\beta_{\zeta}\left(u_{1,i}\right)$
(or $\beta_{\xi}\left(u_{1,i}\right)$) have an asymptotic joint Gaussian
distribution, but, in general, the variances of real and imaginary
parts of such variables are not equal. Interestingly, a closed form
expression for the cumulative density function (cdf) of the product
of $\beta_{\zeta}\left(u_{1,i}\right)$ and $\beta_{\xi}\left(u_{1,i}\right)$
(i.e., the cdf of\ $\frac{1}{MN}\mathbf{a}_{1}^{H}\mathbf{a}_{i}$)
exists in the special case when $\operatorname*{var}\left[\operatorname{Re}\beta_{\zeta}\left(u_{1,i}\right)\right]=\operatorname*{var}\left[\operatorname{Im}\beta_{\zeta}\left(u_{1,i}\right)\right]$
and $\operatorname*{var}\left[\operatorname{Re}\beta_{\xi}\left(u_{1,i}\right)\right]=\operatorname*{var}\left[\operatorname{Im}\beta_{\xi}\left(u_{1,i}\right)\right]$.
By meeting these conditions, in the following theorem we derive an
upper bound on the sidelobes' complementary cdf (ccdf), i.e., $\Pr\left(\frac{1}{MN}\left\vert \mathbf{a}_{1}^{H}\mathbf{a}_{i}\right\vert >q\right)$,
and show that sidelobes have uniformly distributed phases.

We address two MIMO radar setups: (1) $M$ transmitters and $N$ receivers,
where $\xi$ and $\zeta$ are independent, and (2) $N$ transceivers,
where $\xi_{n}=\zeta_{n}$, for all $n$ and $M=N$.

\begin{theorem} \label{Theorem cdf sidelobes}Let the locations $\xi$
of the transmit elements be drawn i.i.d. from a distribution $p\left(\xi\right)$,
and the locations $\zeta$ of the receive elements be drawn i.i.d.
from a distribution $p\left(\zeta\right)$. Assume that\ $p\left(\xi\right)$,
$p\left(\zeta\right)$ and the uniform grid $\phi_{1:G}$ are such
that the transmitter and receiver characteristic functions\ satisfy
\begin{equation}
\psi_{\xi}\left(u_{1,i}\right)=\psi_{\xi}\left(2u_{1,i}\right)=\psi_{\zeta}\left(u_{1,i}\right)=\psi_{\zeta}\left(2u_{1,i}\right)=0\label{eq: eta zero}
\end{equation}
for $i=2,\ldots,G$, where $u_{1,i}=\pi Z\left(\phi_{i}-\phi_{1}\right)$.
Then for $i=2,\ldots,G$:\\
1) If $\xi$ and $\zeta$ are independent:
\begin{equation}
\Pr\left(\frac{1}{MN}\left\vert \mathbf{a}_{1}^{H}\mathbf{a}_{i}\right\vert >q\right)<x\cdot K_{1}\left(x\right),\label{eq: cdf mu}
\end{equation}
where $x\triangleq2\sqrt{MN}q$.\\
2) If $\xi_{n}=\zeta_{n}$ for all $n$:
\begin{equation}
\Pr\left(\frac{1}{N^{2}}\left\vert \mathbf{a}_{1}^{H}\mathbf{a}_{i}\right\vert >q\right)<\exp\left(-Nq\right).\label{eq: cdf mu2}
\end{equation}
3) In both scenarios, the phase of $\mathbf{a}_{1}^{H}\mathbf{a}_{i}$
is uniformly distributed on the unit circle, i.e.,
\begin{equation}
\measuredangle\mathbf{a}_{1}^{H}\mathbf{a}_{i}\sim\mathcal{U}\left[0,2\pi\right].
\end{equation}

\end{theorem}

\begin{proof} See Appendix C. \end{proof}

This theorem characterizes the distribution of $\frac{1}{MN}\mathbf{a}_{1}^{H}\mathbf{a}_{i}$
for the $M$ transmitters $N$ receivers setup, and for the $N$ transceivers
setup. In subsection IV.D, we provide a practical setup that satisfies
(\ref{eq: eta zero}). As shown below, this allows to obtain a uniform
recovery guarantee for spatial compressive sensing.

\subsection{Uniform recovery}

The following corollary of Theorem \ref{Theorem cdf sidelobes} bounds
the probability that the matrix $\mathbf{A}$ has high coherence,
or equivalently, the probability of a peak sidelobe:

\begin{corollary} \label{Corollary Coherence}Let the locations of
the transmit elements $\xi$ be drawn i.i.d. from a distribution $p\left(\xi\right)$,
and the locations of the receivers $\zeta$ be drawn i.i.d. from a
distribution $p\left(\zeta\right)$. Assume that the distributions
$p\left(\xi\right)$ and $p\left(\zeta\right)$ and the uniformly
spaced grid-points $\phi_{1:G}$ are such that (\ref{eq: eta zero})
holds for $i=2,\ldots,G$. Then:\\
1) If $\xi$ and $\zeta$ are independent:
\begin{equation}
\Pr\left(\mu>q\right)<1-\left[1-x\cdot K_{1}\left(x\right)\right]^{G-1},\label{eq: upper bound mu}
\end{equation}
where $x\triangleq2\sqrt{MN}q$.\\
2) If $\xi_{n}=\zeta_{n}$ for all $n$:
\begin{equation}
\Pr\left(\mu>q\right)<1-\left[1-\exp\left(-Nq\right)\right]^{G-1}.\label{eq: upper bound mu2}
\end{equation}

\end{corollary}

\begin{proof} See Appendix D. \end{proof}

Since $\mu$ can be interpreted as the peak sidelobe of the array
pattern, eq. (\ref{eq: upper bound mu}) (eq. (\ref{eq: upper bound mu2}))
characterizes the probability of having a peak sidelobe higher than
$q$ in a system with $M$ transmitters and $N$ receivers ($N$ transceivers).
These results are not asymptotic (i.e., they do not need the number
of measurements $M$ and $N$ to tend to infinity). To further explore
this point, in numerical results we compare these bounds with empirical
simulations.

The coherence $\mu$ plays a major role in obtaining uniform recovery
guarantees for compressive sensing algorithms, as well as guaranteeing
the uniqueness of the sparsest solution to (\ref{eq: L0}). For instance,
using the coherence $\mu$,\ it is possible to obtain a bound on
the RIP constant, $\delta_{K}\leq\left(K-1\right)\mu$ \cite{rauhut}.
This ensures stable and robust recovery by $l_{1}$-minimization (i.e.,
using (\ref{eq: lasso2})) from noisy measurements. By building on
Corollary \ref{Corollary Coherence}, the following theorem establishes
the number of elements $MN$ needed to obtain uniform\ recovery with
high probability using (\ref{eq: lasso2}):

\begin{theorem}[Uniform recovery guarantee]\label{Theorem MN coherence}Let
the locations $\xi$ of the transmit elements be drawn i.i.d. from
a distribution $p\left(\xi\right)$, and the locations $\zeta$ of
the receivers be drawn i.i.d. from a distribution $p\left(\zeta\right)$.
Let the distributions $p\left(\xi\right)$ and $p\left(\zeta\right)$,
and the uniform grid $\phi_{1:G}$ be such that relations (\ref{eq: eta zero})
hold for $i=2,\ldots,G$. Further, let
\begin{equation}
MN\geq C\left(K-\frac{1}{2}\right)^{2}\left[\log\gamma+\frac{1}{2}\log\left(2\log\gamma\right)\right]^{2}\label{eq: MN coherence}
\end{equation}
where $\gamma\triangleq\sqrt{\pi}G/\left(2\epsilon\right)$, and the
constant $C=\left(43+12\sqrt{7}\right)/16\approx4.6718$. Then, with
probability at least $1-\epsilon$, for any $K$-sparse signal $\mathbf{x}\in\mathbb{C}^{G}$
measured from $MN$ MIMO radar measurements\ $\mathbf{y}=\mathbf{Ax}+\mathbf{e}$,
with $\left\Vert \mathbf{e}\right\Vert _{2}\leq\sigma$, the solution
$\mathbf{\hat{x}}$ of (\ref{eq: lasso2}) satisfies
\begin{equation}
\left\Vert \mathbf{\hat{x}-x}\right\Vert _{2}\leq c\sigma,\label{eq: error2}
\end{equation}
where $c$ is a constant that depends only on $\epsilon$. \end{theorem}

\begin{proof} See Appendix E. \end{proof}

The significance of (\ref{eq: MN coherence}) is to indicate the number
of elements necessary to control the peak sidelobe. This is used to
obtain a uniform recovery guarantee for spatial compressive sensing.
In addition, the previous theorem ensures exact recovery of any $K$-sparse
signal using (\ref{eq: lasso2}) in the noise-free case $\sigma=0$.

It is important to point out that the number of grid points $G$ is
not a free variable since $\phi_{1:G}$ must satisfy (\ref{eq: eta zero}).
This point will be explored in subsection IV.D, where we show that
the resolution $G$ must be linearly proportional to the ``virtual\textquotedblright{}\ array
aperture $Z$.

Uniform recovery guarantees capture a worst case recovery scenario.
Indeed, the average performance is usually much better than that predicted
by uniform recovery guarantees. In the following section, we show
that if we consider a non-uniform recovery guarantee, then the zero
mean conditions (\ref{eq: eta zero})\ can be relaxed, and we obtain
recovery guarantees that scale linearly with $K$.

\subsection{Non-uniform recovery}

We now investigate non-uniform recovery guarantees. In recent work
\cite{CandesRIPless}, it has been shown that for a sufficient number
of i.i.d.\ compressive sensing measurements, non-uniform recovery
is guaranteed if isotropy holds. However, the result in \cite{CandesRIPless}
cannot be directly used in our framework since the $MN$ rows of the
matrix $\mathbf{A}$, following (\ref{eq: ateta}), are not i.i.d.
This scenario is addressed in \cite{RahutStromer} in which non-uniform
recovery is guaranteed for a MIMO radar system with $N$ transceivers
if the isotropy property (under the name \textit{aperture condition})
holds. The following theorem derives conditions on grid points $\phi_{1:G}$
and probability distributions $p\left(\xi\right)$ and $p\left(\zeta\right)$,
in order for the random matrix $\mathbf{A}$ to satisfy the isotropy
property:

\begin{theorem} \label{Lemma Isotropy}Let the locations $\xi$ of
the transmit elements be drawn i.i.d. from a distribution $p\left(\xi\right)$,
and the locations $\zeta$ of the receivers be drawn i.i.d. from a
distribution $p\left(\zeta\right)$. For every $t$, the $t$-th row
of $\mathbf{A}$ in (\ref{eq: y}) satisfies the isotropy property
\cite{CandesRIPless}, i.e.,
\begin{equation}
\mathbb{E}\left[\mathbf{A}^{H}\left(t,:\right)\mathbf{A}\left(t,:\right)\right]=\mathbf{I,}
\end{equation}
iff $p\left(\xi\right)$, $p\left(\zeta\right)$ and $\phi_{1:G}$
are chosen such that, for $i=2,\ldots,G$, 
\begin{equation}
\psi_{z}\left(u_{1,i}\right)=0,\label{eq: dirac}
\end{equation}
where $z\triangleq\zeta+\xi$ and $u_{1,i}\triangleq\pi Z\left(\phi_{i}-\phi_{1}\right)$.
\end{theorem}

\begin{proof} See Appendix F. \end{proof}

Theorem \ref{Lemma Isotropy} links grid points $\phi_{1:G}$ and
probability distributions $p\left(\xi\right)$ and $p\left(\zeta\right)$
(through the characteristic function of $z$) with the isotropy property
of $\mathbf{A}$. When (\ref{eq: dirac}) holds, it can be shown that
the aperture condition used in \cite{RahutStromer} holds too. Therefore,
using the same approach as in \cite{RahutStromer}, non-uniform recovery
of\ $K$ targets via (\ref{eq: lasso2}) is guaranteed in the proposed
spatial compressive sensing framework. The following Theorem customizes\ Theorem
2.1 in \cite{RahutStromer} to our framework:

\begin{theorem}[Non-uniform recovery guarantee]\label{theorem Rauhut}Consider
a $K$-sparse $\mathbf{x}\in\mathbb{C}^{G}$ measured from $MN$ MIMO
radar measurements\ $\mathbf{y}=\mathbf{Ax}+\mathbf{e}$\textbf{,}
where $\left\Vert \mathbf{e}\right\Vert _{2}\leq\sigma$. Let $\varepsilon>0$
be an arbitrary scalar, and suppose that the random matrix $\mathbf{A}$
satisfies the isotropy property, $\mathbb{E}\left[\mathbf{A}^{H}\left(t,:\right)\mathbf{A}\left(t,:\right)\right]=\mathbf{I}$
$\forall t$. Then with probability at least $1-\varepsilon$, the
solution $\mathbf{\hat{x}}$ to (\ref{eq: lasso2}) obeys
\begin{equation}
\left\Vert \mathbf{\hat{x}-x}\right\Vert _{2}\leq C_{1}\sigma\sqrt{\frac{K}{MN}},\label{eq: error}
\end{equation}
provided that the number of rows of $\mathbf{A}$ meets
\begin{equation}
MN\geq CK\log^{2}\left(\frac{cG}{\varepsilon}\right),\label{eq: MN_beta}
\end{equation}
where $C_{1}$, $C$ and $c$ are constants. \end{theorem}

\begin{proof} The theorem results from Theorem 2.1 in \cite{RahutStromer}
by performing the following substitutions: $K$ for $s$ (sparsity),
$MN$ for $n^{2}$ (number of rows of $\mathbf{A}$), and $G$ for
$N$ (number of columns of $\mathbf{A}$). Since in this work we consider
$K$-sparse signals, in (\ref{eq: error}) we discarded the term that
accounts for\ nearly-sparse signals present in \cite{RahutStromer}.
\end{proof}

Theorem \ref{theorem Rauhut} shows that, when the isotropy property
is satisfied, the proposed framework enables us to localize $K$ targets
using about $MN=K\left(\log G\right)^{2}$ MIMO radar measurements.

Some comments are in order. First, it is important to stress that
in (\ref{eq: MN_beta}), the number of elements scales linearly with
the sparsity $K$. This is in contrast with uniform recovery bounds
based on coherence (e.g. (\ref{eq: MN coherence})), which scale quadratically
with $K$. Moreover, the significance of the logarithmic dependence
on $G$ is that the proposed framework enables high resolution with
a small number of MIMO radar elements. This is in contrast with a
filled virtual MIMO\ array where the product $MN$ scales linearly
with $G$. Again, it is crucial to point out that the number of grid
points $G$ is not a free variable, because the grid points $\phi_{1:G}$
must satisfy (\ref{eq: dirac}). Second, differently from (\ref{eq: error}),
in (\ref{eq: error2}) the error did not depend on $K$, $M$ and
$N$. Third, (\ref{eq: error}) shows that reconstruction is stable
even when the measurements are noisy. Additionally, we see from (\ref{eq: error})
that when $\sigma=0$, Theorem \ref{theorem Rauhut} guarantees exact
reconstruction with high probability, when (\ref{eq: MN_beta}) holds.
Both results above can be extended to approximately sparse vectors,
in which case an extra term appears in the right hand-side of (\ref{eq: error2})
and (\ref{eq: error}). This situation may emerge when targets are
not exactly on a grid, however, the analysis of such scenario is outside
the scope of this paper. Finally, to suggest some intuition into the
above conditions, notice that recovery can be guaranteed by requiring
the matrix $\mathbf{A}$ to satisfy the isotropy property, $\mathbb{E}\left[\mathbf{Q}\right]=MN\cdot\mathbf{I}$,
and by controlling the variances of the non-diagonal elements of $\mathbf{Q}$(which,
according to (\ref{e:sig1}) and (\ref{e:sig2}), scale with $1/MN$)
through the use of a sufficient number of measurements $MN$.

\subsection{Element locations and grid-points}

We now provide an example of $p\left(\xi\right)$, $p\left(\zeta\right)$
and $\phi_{1:G}$ that meet the requirements of Theorem \ref{Theorem cdf sidelobes}
and Theorem \ref{Lemma Isotropy}.

The conditions needed by each theorem constraint the characteristic
function of the random variables $\xi$, $\zeta$. Let, $Z_{TX}=Z_{RX}=Z/2$,
such that the random variables $\xi$ and $\zeta$ are both confined
to the interval $\left[-\frac{1}{2},\frac{1}{2}\right]$. The characteristic
function of a uniform random variable $\zeta\sim\mathcal{U}\left[-\frac{1}{2},\frac{1}{2}\right]$
is the sinc function, i.e.,
\begin{equation}
\psi_{\zeta}\left(u\right)=\frac{\sin\left(u/2\right)}{u/2}.\label{eq: sinc square}
\end{equation}
Therefore, when $\zeta$ is uniformly distributed, by choosing $\phi_{1:G}$
as a uniform grid of $2/Z$-spaced points in the range $\left[-1,1\right]$,
we have that $\psi_{\zeta}\left(u_{i,l}\right)=\psi_{\zeta}\left(2u_{i,l}\right)=0$
for any $i\neq l$ (since $u_{i,l}\triangleq\pi Z\left(\phi_{l}-\phi_{i}\right)=2\pi\left\vert i-l\right\vert $).
The number of grid points $G$ is not a free variable, because the
grid points $\phi_{1:G}$ must satisfy (\ref{eq: eta zero}) or (\ref{eq: dirac}).
For instance, in the example above, $\phi_{1:G}$ must be a uniform
grid of $2/Z$-spaced points between $\left[-1,1\right]$, and, assuming
that $Z$ is an integer, the number of grid points is $G=Z+1$.

The dependence between the number of grid points $G$ and the virtual
array aperture $Z$ can be understood by noticing that both (\ref{eq: eta zero})
and (\ref{eq: dirac}) impose that grid points are placed at the zeros
of the characteristic function of the relative random variable (i.e.
the sinc function). The spacing of the zeros is dictated by the virtual
array aperture $Z$. The bigger the aperture the more grid points
fit in the range $\left[-1,1\right]$.

Summarizing, choose $\phi_{1:G}$ as a uniform grid of $2/Z$-spaced
points in the range $\left[-1,1\right]$. Then:
\begin{enumerate}
\item If \textit{both} $\zeta$ and $\xi$ are uniformly distributed, relations
(\ref{eq: eta zero}) hold, and we can invoke Theorem \ref{Theorem cdf sidelobes}
(for uniform recovery);
\item If \textit{either} $\zeta$ or $\xi$ are uniformly distributed, relation
(\ref{eq: dirac}) holds, and we can invoke Theorem \ref{Lemma Isotropy}
(for non-uniform recovery). 
\end{enumerate}
Note that non-uniform recovery, i.e., (\ref{eq: dirac}), requires
only one density function, say $p\left(\zeta\right)$, to be uniform,
while the other distribution, $p\left(\xi\right)$, can be arbitrarily
chosen, e.g., it can be even deterministically dependent on $\zeta$.
For instance, (\ref{eq: dirac}) is satisfied in a MIMO radar system
with $N$ transceivers, i.e., when $\zeta_{1:N}$ are i.i.d. uniform
distributed and we deterministically set $\xi_{n}=\zeta_{n}$.

Finally, we remark that the analysis provided in this section regarding
the statistics of the matrix $\mathbf{A}$ may be used with block-sparsity
results in the compressive sensing literature \cite{EldarCS} to obtain
guarantees for the general MMV scenario.

\section{Numerical Results\label{Sec_Numerical}}

In this section, we present numerical results illustrating the proposed
spatial compressive sensing framework.%TCIMACRO{\FRAME{ftbpFU}{3.4402in}{2.7544in}{0pt}{\Qcb{Empirical ccdf of the
%coherence of the measurment matrix $\QTR{bf}{A}$ and its upper bound as a
%function of the number of elements. (a) considers the $M$ transmitters and $N$
%receivers setup and the upper bound is given in (\ref{eq: upper bound mu});
%(b) considers the $N$ tranceivers setup and the upper bound is given in
%(\ref{eq: upper bound mu2}).}}{\Qlb{fig: ccdf mu}}%
%{hist_coherence_mn_tranc_g251.eps}{\special{ language "Scientific Word";
%type "GRAPHIC";  maintain-aspect-ratio TRUE;  display "USEDEF";
%valid_file "F";  width 3.4402in;  height 2.7544in;  depth 0pt;
%original-width 5.0704in;  original-height 4.0499in;  cropleft "0";
%filename 'graphics/hist_coherence_MN_tranc_G251.eps';file-properties "XNPEU";}%
%}}%
%BeginExpansion
\begin{figure}
\centering{}\includegraphics[width=3.4402in,height=2.7544in]{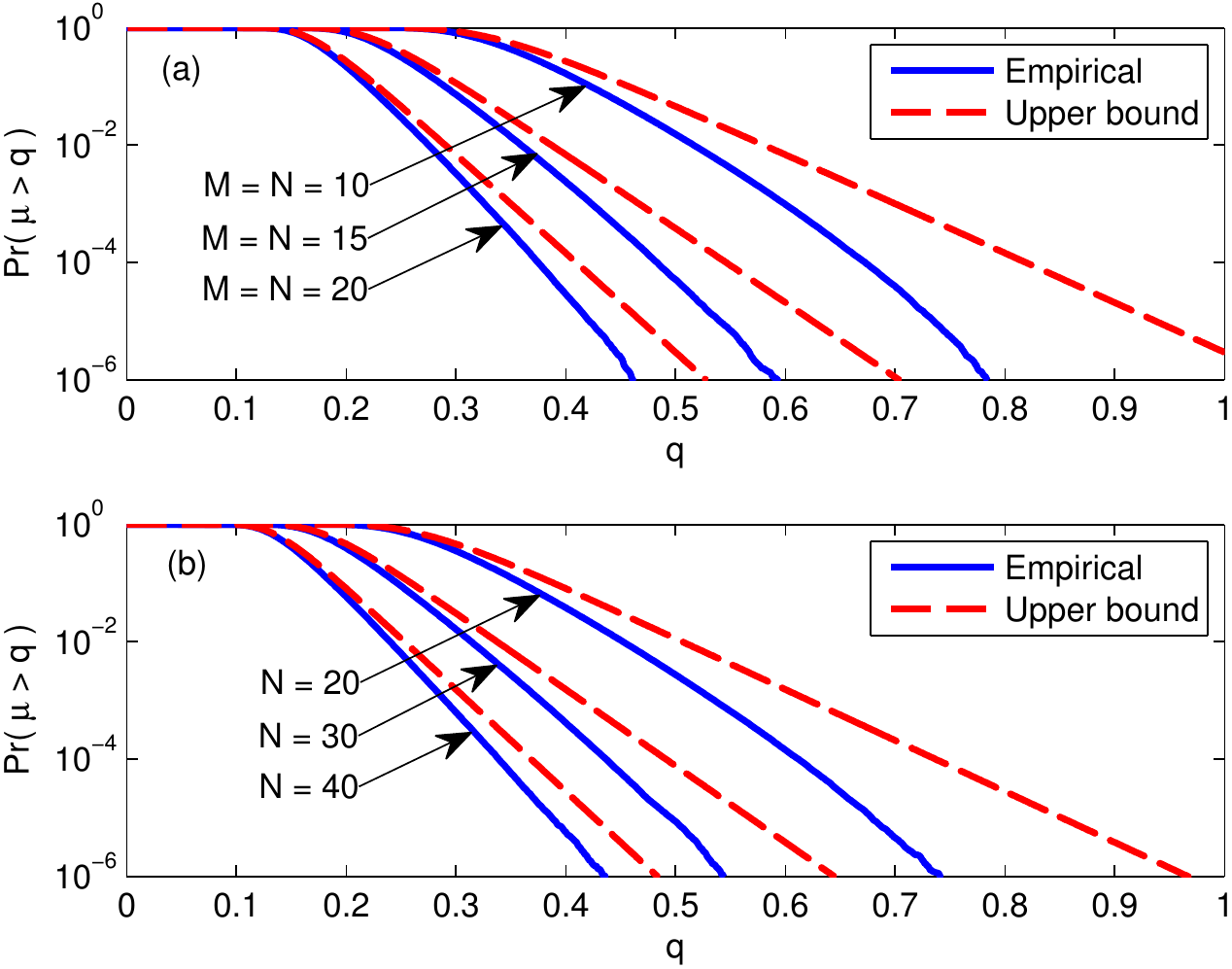}\caption{Empirical ccdf of the coherence of the measurement matrix $\mathbf{A}$
and its upper bound as a function of the number of elements. (a) considers
the $M$ transmitters and $N$ receivers setup and the upper bound
is given in (\ref{eq: upper bound mu}); (b) considers the $N$ transceivers
setup and the upper bound is given in (\ref{eq: upper bound mu2}).}
\label{fig: ccdf mu}
\end{figure}

%EndExpansion

We design an example to follow Theorem \ref{Theorem cdf sidelobes},
in which $p\left(\xi\right)$ and $p\left(\zeta\right)$ are both
uniform distributions, and $\phi_{1:G}$ represents a uniform grid
of $2/Z$-spaced points in the interval $\left[-1,1\right]$, which
implies that the number of grid points is $G=Z+1$. The system transmits
a total of $P$ pulses. When expressed in discrete form, each pulse
consists of $M$\ orthogonal codes composed by\ $M$\ symbols.
In particular, we select the codes to be the rows of the $M\times M$\ Fourier
matrix. Equal length apertures were assumed for the transmit and receive
arrays, i.e., $Z_{TX}=Z_{RX}=Z/2$. The target gains were given by
$x_{k}=\exp\left(-j\varphi_{k}\right)$, with $\varphi_{k}$ drawn
i.i.d., uniform over $\left[0,2\pi\right)$, for all $k=1,...,K$
(where $K$ is the number of targets). The noise (see (\ref{eq: y}))
was assumed to be distributed as $\operatorname{vec}\left(\mathbf{E}\right)\sim\mathcal{CN}\left(\mathbf{0},\sigma^{2}\mathbf{I}\right)$
and the SNR is defined as $\mathbf{-}10\log_{10}\sigma^{2}$. From
the definition of the measurement matrix $\mathbf{A}$, its columns
all have squared-norm equal to $MN$. Throughout the numerical results,
we normalize the columns of $\mathbf{A}$ to have unit norm.

We first investigate the statistics of the matrix $\mathbf{Q}$ discussed
in Section \ref{Sec_spatialCS}. In particular, we analyze the coherence
$\mu$ of the measurement matrix $\mathbf{A}$ compared to the result
given in Theorem \ref{Theorem cdf sidelobes}. The virtual aperture
was $Z=250$ (thus $G=251$). In Fig. \ref{fig: ccdf mu}, we plot
the ccdf of the coherence $\mu$, i.e. $\Pr\left(\mu>q\right)$, as
a function of the number of elements for (a) the $M$ transmitters
and $N$ receivers setup and (b) the $N$ transceivers setup. As a
reference, we also plot the upper bound given in (\ref{eq: upper bound mu})
and (\ref{eq: upper bound mu2}), respectively.\ It can be seen how
the upper bound becomes tighter and tighter as the number of elements
increases. In addition, it is interesting to notice that the coherence
of the matrix $\mathbf{A}$ for the $N^{\prime}$ transceivers setup
is very close to the coherence of the matrix $\mathbf{A}$ for the
the $M$ transmitters and $N$ receivers setup when $M=N=N^{\prime}/2$.

We next present localization performance using practical algorithms.
We implemented target localization using LASSO following the algorithm
proposed in \cite{RahutStromer} to solve problem (\ref{eq: lasso2}).
In addition, we implement Beamforming, OLS, OMP, CoSaMP, FOCUSS \cite{focuss}
and MBMP. In the MMV setup we also compare MBMP, RA-ORMP \cite{Davies10},
M-FOCUSS \cite{focuss}, and MUSIC \cite{Schmidt86}. Concerning MBMP,
it requires as input a $K$ length branch vector $\mathbf{d}$, which
define the algorithm's complexity (see \cite{Rossi12} for details
on setting parameters for MBMP). The output of MBMP is the estimated
support. Notice that, when $\mathbf{d}=\left[1,\ldots,1\right]$,
MBMP reduces to OLS in the SMV scenario, and to RA-ORMP in the MMV
scenario. We define the support recovery error when the estimated
support does not coincide with the true one. For algorithms that return
an estimate $\mathbf{\hat{x}}$ of the sparse vector $\mathbf{x}$
(e.g., LASSO, FOCUSS and MUSIC), the support was then identified as
the $K$ largest modulo entries of the signal $\mathbf{\hat{x}}$.

Analyzing how to set the noise parameter $\sigma$ in (\ref{eq: L0}),
or the sparsity $K$, without prior information is the topic of current
work \cite{Rossi12b}, but outside the scope of this paper. Therefore,
we assume that the noise level is available and that the number of
targets $K$ is known (notice that this information is needed by all
the algorithms including MBMP). The virtual aperture was $Z=250$
(thus $G=251$), and tests were carried out for $K=5$ targets. The
SNR was $20$ dB throughout.

The main focus of the paper is to reduce the number of antenna elements
while avoiding sidelobes errors and while preserving the high-resolution
provided by the virtual array aperture $Z$ (i.e., recovering $2/Z$-spaced
targets). Therefore, to account for errors due to sidelobes (an erroneous
target is estimated at a sidelobe location) and unresolved targets
(the responses of two targets in consecutive grid-points is merged
in only one grid-point), we consider as performance metric the support
recovery error probability, defined as the error event when at least
one target is estimated erroneously.

We first treat the non-uniform guarantee setting. Monte Carlo simulations
were carried out using independent realizations of target gains, targets
locations, noise and element positions. Fig. \ref{fig: nonuniform}
illustrates the probability of support recovery error as a function
of the number of measurements $MN$. From the figure, it can be seen
that compressive sensing algorithms enable better performance (smaller
probability of sidelobe error and better resolution) than beamforming,
which is not well-suited for the sparse recovery framework. Among
compressive sensing algorithms, two main groups appear: on one side,
OLS and OMP, which both have practically the same performance; on
the other side, L1-norm opt., CoSaMP, FOCUSS and MBMP. Among the latter
group, it is important to point out that, although the recovery guarantee
established in Theorem \ref{theorem Rauhut} requires the solution
of (\ref{eq: lasso2}), and thus using LASSO, MBMP provides a viable
and competitive way to perform target localization.%TCIMACRO{\FRAME{ftbpFU}{3.4402in}{2.655in}{0pt}{\Qcb{Probability of support
%recovery\ error as a function of the number of rows $MN$ of $\QTR{bf}{A}$.
%Non-uniform SMV setup. The system settings are $Z=250$, $G=251$, $P=1$ and
%$K=5$ targets with $\left\vert x_{k}\right\vert =1$ for all $k$. The SNR is
%$20$db.}}{\Qlb{fig: nonuniform}}%
%{doa_nonuniform_smv_snr20_k5_p1_g251_mn_vec.eps}%
%width 3.4402in;  height 2.655in;  depth 0pt;  original-width 5.0825in;
%original-height 4.0101in;  cropleft "0";  croptop "1";  cropright "1";
%filename 'graphics/DOA_nonuniform_SMV_snr20_K5_P1_G251_MN_vec.eps';file-properties "XNPEU";}%
%}}%
%BeginExpansion
\begin{figure}[ptb]
\centering{}\includegraphics[width=3.4402in,height=2.655in]{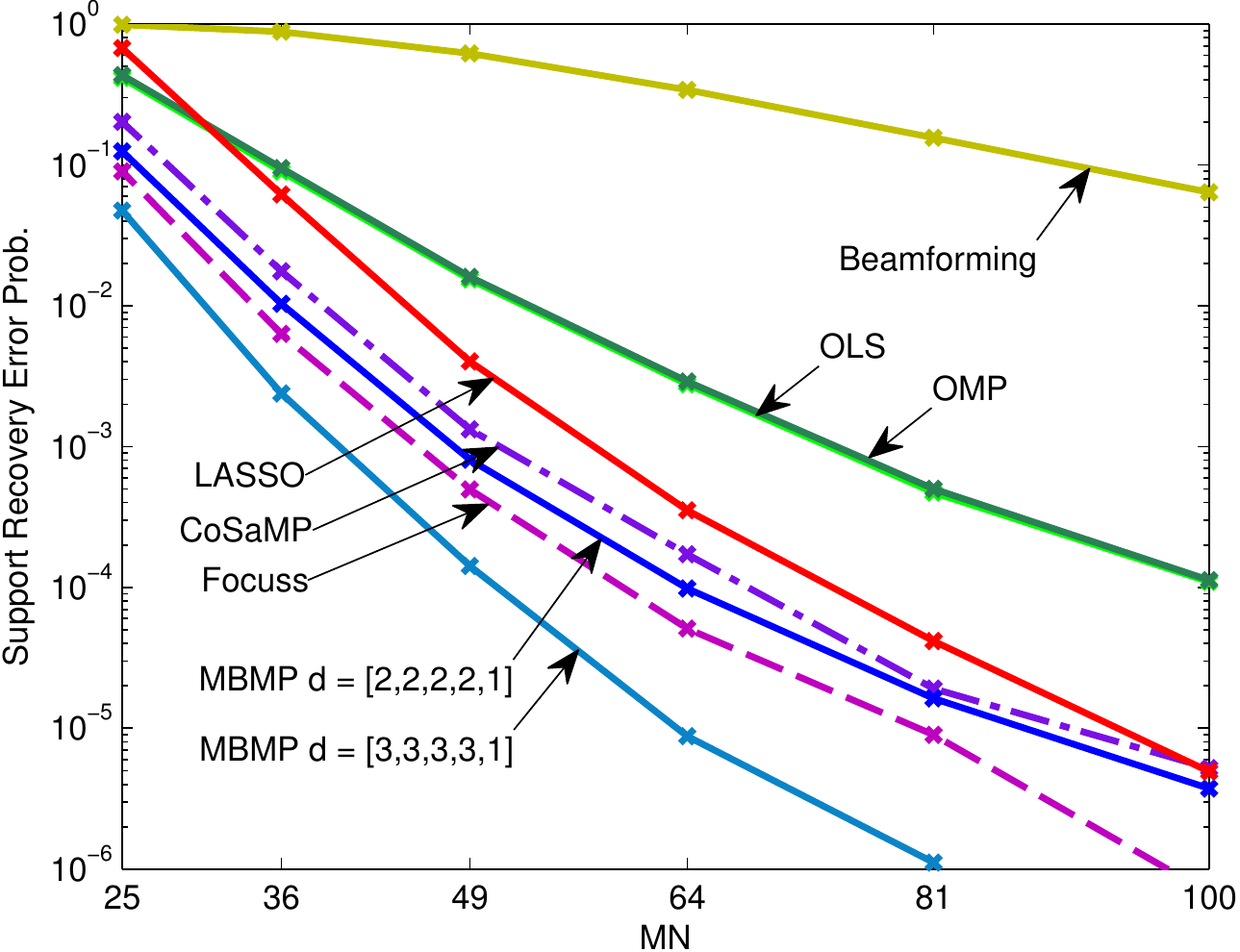}\caption{Probability of support recovery error as a function of the number
of rows $MN$ of $\mathbf{A}$. Non-uniform SMV setup. The system
settings are $Z=250$, $G=251$, $P=1$ and $K=5$ targets with $\left\vert x_{k}\right\vert =1$
for all $k$. The SNR is $20$db.}
\label{fig: nonuniform}
\end{figure}

%EndExpansion%
%TCIMACRO{\FRAME{ftbpFU}{3.4402in}{2.7121in}{0pt}{\Qcb{Probability of support
%recovery\ error as a function of the number of rows $MN$ of $\QTR{bf}{A}$.
%Uniform SMV setup. The system settings are $Z=250$, $G=251$, $P=1$ and $K=5$
%targets with $\left\vert x_{k}\right\vert =1$ for all $k$. The SNR is
%$20$db.}}{\Qlb{fig: uniform}}{doa_uniform_smv_snr20_k5_p1_g251_mn_vec.eps}%
%width 3.4402in;  height 2.7121in;  depth 0pt;  original-width 5.0825in;
%original-height 4.0101in;  cropleft "0";  croptop "1";  cropright "1";
%filename 'graphics/DOA_uniform_SMV_snr20_K5_P1_G251_MN_vec.eps';file-properties "XNPEU";}%
%}}%
%BeginExpansion
\begin{figure}[ptb]
\centering{}\includegraphics[width=3.4402in,height=2.7121in]{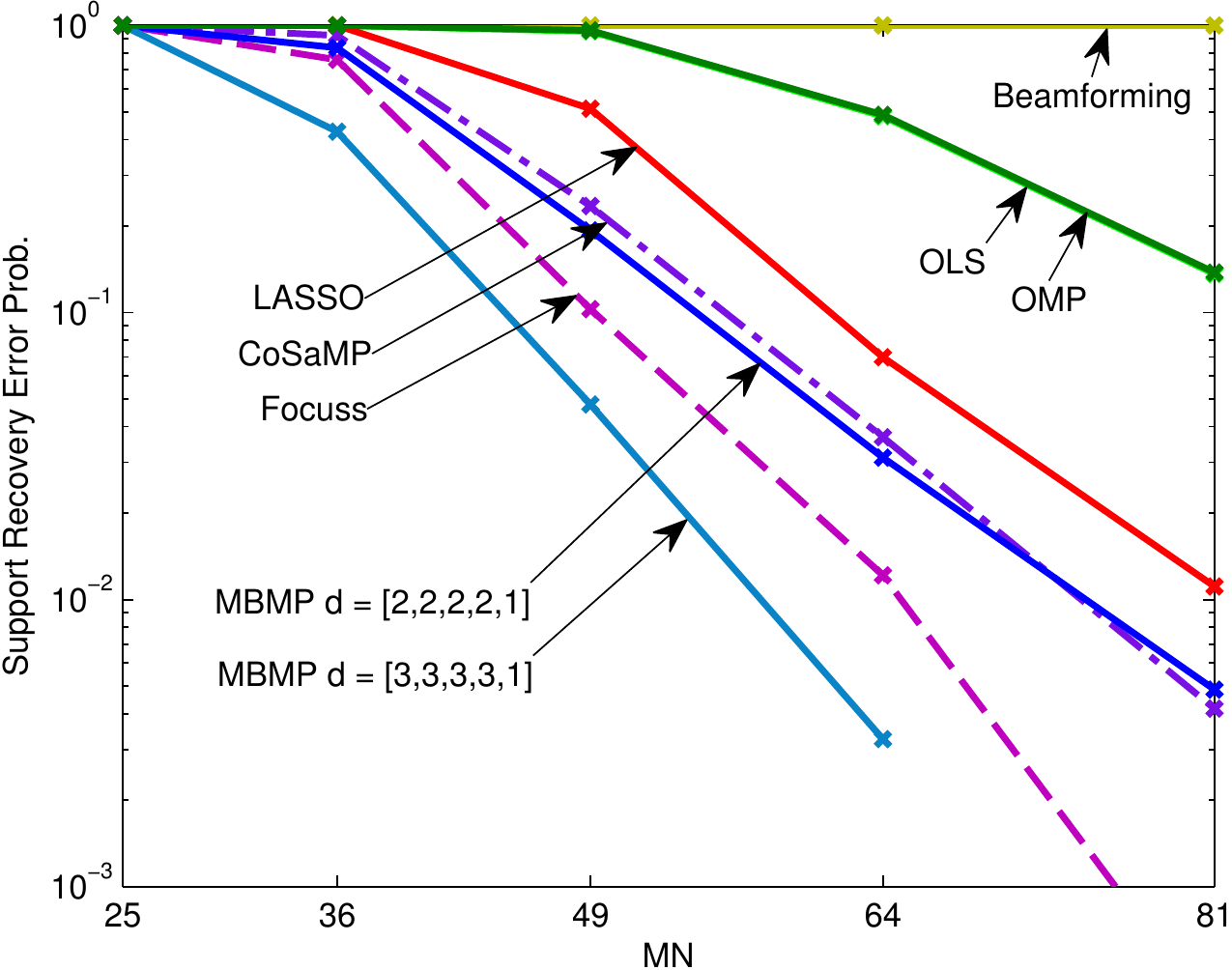}\caption{Probability of support recovery error as a function of the number
of rows $MN$ of $\mathbf{A}$. Uniform SMV setup. The system settings
are $Z=250$, $G=251$, $P=1$ and $K=5$ targets with $\left\vert x_{k}\right\vert =1$
for all $k$. The SNR is $20$db.}
\label{fig: uniform}
\end{figure}

%EndExpansion

We next consider uniform guarantees. In this setup, we first generate
a realization of the matrix $\mathbf{A}$ by drawing at random the
element positions. Maintaining the matrix $\mathbf{A}$ fixed, we
perform $500$ Monte Carlo simulations using independent realizations
of target gains, targets locations and noise. For each recovery method,
we defined a support recovery error if an error occur in any of the
$500$ simulations. We then average throughout element positions realizations.\ Fig.
\ref{fig: uniform} illustrates the probability of support recovery
error as a function of the number of measurements $MN$. The difference
among OLS/OMP and the more sophisticated methods (i.e., LASSO, CoSaMP,
FOCUSS and MBMP) is even more evident in this setup (e.g. at $MN=81$,
the probability of OLS/OMP error is greater than $0.1$), confirming
the theoretical finding \cite{RauhutNouniformOMP} of OLS/OMP unfitness
to deliver uniform recovery. On the other hand, MBMP, an extension
of OLS, still provides competitive performance. In particular, MBMP
with $\mathbf{d}=\left[3,3,3,3,1\right]$ outperform the other methods.

The theoretical results presented in this work focus on the SMV setting.
However, in practice several snapshots can be available. To explore
the benefits of the proposed MIMO random array framework in such case,
in Fig. \ref{fig: P5 snr20 MNvec} we consider an MMV setting ($P=5$)
and we compare sparse recovery methods with the well-known MUSIC algorithm.
We evaluate five different elements configurations: $\left[M,N\right]=\left[3,3\right]$,
$\left[4,4\right]$, $\left[5,5\right]$, $\left[6,6\right]$ and
$\left[7,7\right]$. The figure illustrates the probability of support
recovery error as a function of the number of measurements $MN$ (nonuniform
setup). Sparse recovery algorithms have better performances than MUSIC,
and the availability of multiple snapshots allows to considerably
reduce the number of antenna elements. Moreover, in the MMV setting,
algorithms which are able to exploit the signal subspace information
(e.g. MBMP and RA-ORMP) posses a clear advantage over those algorithms
that are unable (e.g. M-FOCUSS). For instance, this can be appreciated
by the difference in performance of FOCUSS and MBMP with $\mathbf{d}=\left[2,2,2,2,1\right]$
when comparing the SMV (Fig. \ref{fig: nonuniform}) and MMV (Fig.
\ref{fig: P5 snr20 MNvec}) settings. The numerical simulations presented
in this paper considered a medium SNR level and show a superior performance
of sparse recovery methods over classical methods (e.g. beamforming
or MUSIC) in the proposed framework. Since the sparsity property,
upon which sparse recovery methods rely, is independent from the SNR,
we expect a similar behavior also at low SNR (e.g., SNR$=0$ dB or
lower). %TCIMACRO{\FRAME{ftbpFU}{3.4904in}{2.7268in}{0pt}{\Qcb{Probability of support
%recovery\ error as a function of the number of rows $MN$ of $\QTR{bf}{A}$.
%Non-uniform MMV setup. The system settings are $Z=250$, $G=251$, $P=5$ and
%$K=5$ targets with $\left\vert x_{k,p}\right\vert =1$ for all $k,p$. The SNR
%is $20$db.}}{\Qlb{fig: P5 snr20 MNvec}}%
%{doa_nonuniform_mmv_snr20_k5_p5_g251_mn_vec.eps}%
%width 3.4904in;  height 2.7268in;  depth 0pt;  original-width 5.0574in;
%original-height 3.9401in;  cropleft "0";  croptop "1";  cropright "1";
%filename 'graphics/DOA_nonuniform_MMV_snr20_K5_P5_G251_MN_vec.eps';file-properties "XNPEU";}%
%}}%
%BeginExpansion
\begin{figure}[ptb]
\centering{}\includegraphics[width=3.4904in,height=2.7268in]{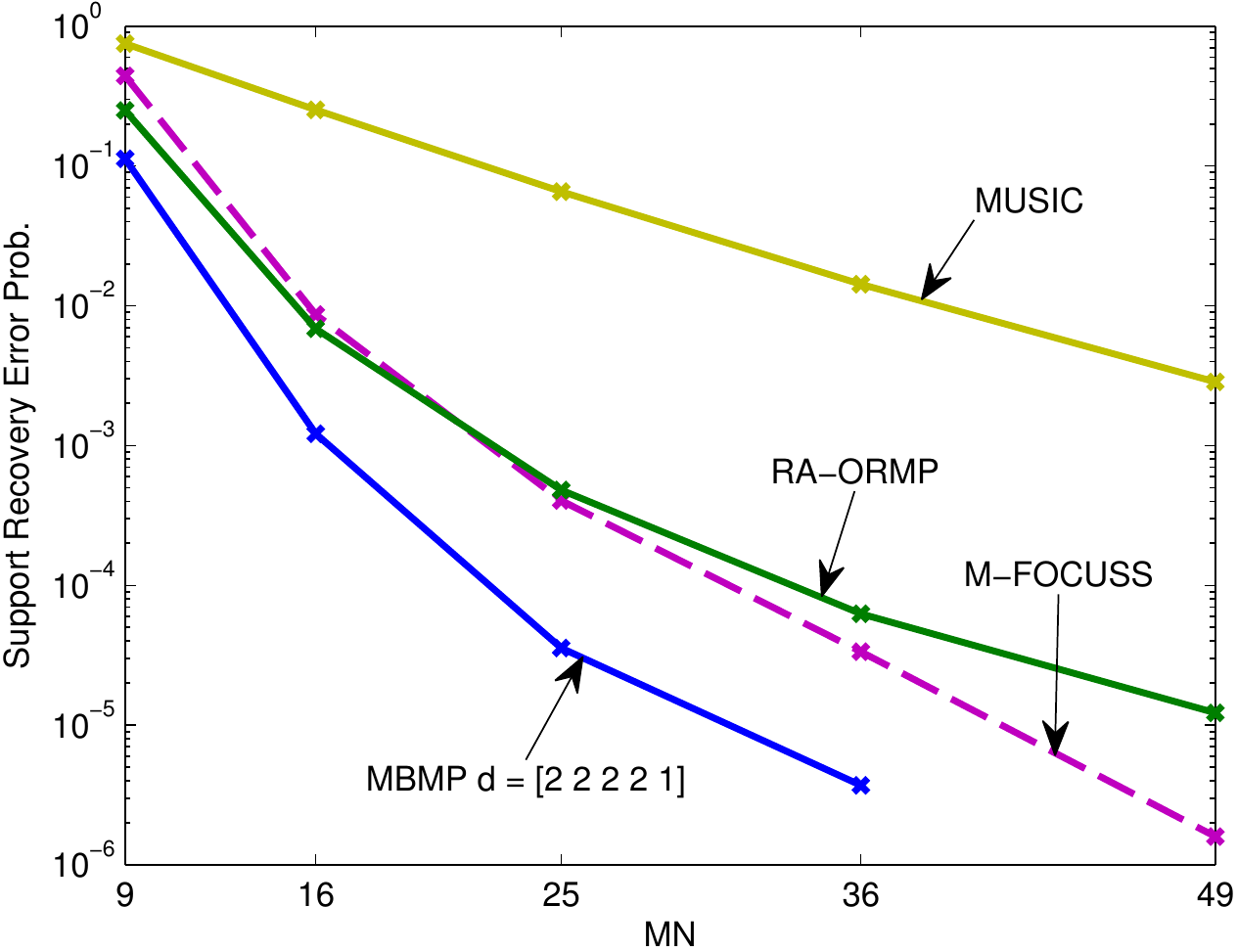}\caption{Probability of support recovery error as a function of the number
of rows $MN$ of $\mathbf{A}$. Non-uniform MMV setup. The system
settings are $Z=250$, $G=251$, $P=5$ and $K=5$ targets with $\left\vert x_{k,p}\right\vert =1$
for all $k,p$. The SNR is $20$db.}
\label{fig: P5 snr20 MNvec}
\end{figure}

%EndExpansion

\section{Conclusions\label{s:conc}}

We propose a sparse framework to address the source localization problem
for a random array MIMO radar system. We link system design quantities,
i.e., the probability distributions $p\left(\xi\right)$ and $p\left(\zeta\right)$
of the tx/rx sensors location and the sparse localization grid points
$\phi_{1:G}$, with the statistics of the Gram matrix $\mathbf{Q}$
and the related coherence of the matrix $\mathbf{A}$. Based on this
result, we were able to develop uniform and non-uniform recovery guarantees
for spatial compressive sensing. We show that within the proposed
framework, it is possible to localize $K$ targets using about $MN=K\left(\log G\right)^{2}$
MIMO radar noisy measurements, where $G$ is proportional to the array
aperture and determines the angle resolution. In other words, the
proposed framework supports the high-resolution provided by the virtual
array aperture while using a reduced number of MIMO radar elements.
This is in contrast with a filled virtual MIMO\ array for which the
product $MN$ scales linearly with $G$. Moreover, since the results
characterize the product of the number of transmit and receive elements,
MIMO random array implementation further reduces the total number
of antenna elements needed. From numerical simulations it emerges
that, in the proposed framework, compressive sensing recovery algorithms
(e.g. MBMP) are capable of better performance (i.e., smaller probability
of sidelobe errors and better resolution) than classical methods,
such as beamforming and MUSIC.

\section{Appendix}

\subsection{Proof of Proposition \ref{Lemma beampattern}\label{App Beampattern}}

\subsubsection{Mean}

The mean $\eta\left(u\right)$ is by definition the expectation of
the random array pattern, i.e., $\mathbb{E}\left[\beta\left(u\right)\right]$,
over $z_{mn}=\xi_{m}+\zeta_{n}$. The expectation and the summations
can be interchanged obtaining
\begin{equation}
\eta\left(u\right)=\frac{1}{MN}\sum_{m=1}^{M}\sum_{n=1}^{N}\mathbb{E}\left[\exp\left(juz_{mn}\right)\right].
\end{equation}
Moreover, the average of $\exp\left[ju_{i,l}\left(\zeta_{n}+\xi_{m}\right)\right]$
does not depend on the index $n$ and $m$, since $\zeta_{1:N}$ are
identically distributed, and so are $\xi_{1:M}$. By dropping the
indexes of $\zeta_{n}$ and $\xi_{m}$ and using $z=\xi+\zeta$, we
have the sum of $MN$ identical terms, divided by $MN$. Thus $\eta\left(u\right)$
equals $\mathbb{E}\left[\exp\left(juz\right)\right]$, the characteristic
function of the random variable\ $z$.

\subsubsection{Variance}

Let $\sigma_{1}^{2}\left(u\right)=$ $\operatorname{Re}\beta\left(u\right)$
and $\sigma_{2}^{2}\left(u\right)=$ $\operatorname{Im}\beta\left(u\right)$.
For brevity of notation, we drop the dependency on $u$. First notice
that since $p\left(z\right)$ is even, its characteristic function
is real, thus so is the mean value of the array pattern $\eta$. We
also have that $\mathbb{E}\left[\left(\operatorname{Re}\beta-\eta\right)\operatorname{Im}\beta\right]=\mathbb{E}\left[\operatorname{Re}\beta\operatorname{Im}\beta\right]-\eta\mathbb{E}\left[\operatorname{Im}\beta\right]=0$,
since the real and imaginary parts are uncorrelated and because $\mathbb{E}\left[\operatorname{Im}\beta\right]=0$.
Next, we need to evaluate $\sigma_{1}^{2}\triangleq\mathbb{E}\left[\left(\operatorname{Re}\beta-\eta\right)^{2}\right]$
and $\sigma_{2}^{2}\triangleq\mathbb{E}\left[\left(\operatorname{Im}\beta\right)^{2}\right]$.
In order to derive these quantities, we consider the expectations
given by $\mathbb{E}\left[\left(\beta-\eta\right)^{2}\right]$ and
$\mathbb{E}\left[|\beta-\eta|^{2}\right]$. It can be shown that,
\begin{equation}
\mathbb{E}\left[\left(\beta-\eta\right)^{2}\right]=\sigma_{1}^{2}-\sigma_{2}^{2}+j2\sigma_{12}\label{eqn:cov1}
\end{equation}
and 
\begin{equation}
\mathbb{E}\left[|\beta-\eta|^{2}\right]=\sigma_{1}^{2}+\sigma_{2}^{2}.\label{eqn:cov2}
\end{equation}
Substituting the definition of the random array pattern $\beta\left(u\right)$
(\ref{eq: beta}) and (\ref{eqn:mean_pat_app}) in (\ref{eqn:cov1})
and (\ref{eqn:cov2}), we obtain (\ref{e:sig1}) and (\ref{e:sig2}).

\subsection{Proof of Lemma \ref{Lemma Toeplitz}\label{App Toeplitz}}

From (\ref{eq: beta}) we have that $\mathbf{a}_{i}^{H}\mathbf{a}_{l}=MN\cdot\beta\left(u_{i,l}\right)$,
where $u_{i,l}\triangleq\pi Z\left(\phi_{i}-\phi_{l}\right)$. When
$\phi_{1:G}$ is a uniform grid, $\phi_{i}-\phi_{l}$ is constant
whenever $i-l$ is constant, i.e., along every diagonal of the matrix
$\mathbf{Q}$. Since $\beta\left(u_{i,l}\right)$ depends only on
the term $\phi_{i}-\phi_{l}$ (not on the actual $\phi_{i}$ and $\phi_{j}$),
$\mathbf{Q}$ is a Toeplitz matrix.

\subsection{Proof of Theorem \ref{Theorem cdf sidelobes}\label{App CDF sidelobes}}

We define the array pattern associated with the transmitter as
\begin{equation}
\beta_{\zeta}\left(u_{i,l}\right)\triangleq\frac{1}{N}\sum_{n=1}^{N}\exp\left[ju_{i,l}\zeta_{n}\right]=\frac{1}{N}\mathbf{b}_{i}^{H}\mathbf{b}_{l}
\end{equation}
and with the receiver arrays as:
\begin{equation}
\beta_{\xi}\left(u_{i,l}\right)\triangleq\frac{1}{M}\sum_{m=1}^{M}\exp\left[ju_{i,l}\xi_{m}\right]=\frac{1}{M}\mathbf{c}_{i}^{H}\mathbf{c}_{l}.
\end{equation}
Statistical properties of random arrays were analyzed in \cite{Lo64}\ in\ the
case of passive localization (i.e., an array with only receiving elements).
The following lemma customizes useful results from \cite{Lo64}:

\begin{lemma} Let the locations $\zeta_{1:N}$ of the receiving array
be i.i.d., drawn from an even distribution $p\left(\zeta\right)$
and consider a given $u$. Then $\beta_{\zeta}\left(u\right)$ is
asymptotically jointly Gaussian distributed (we neglect the dependency
on $u$):
\begin{equation}
\left[\begin{array}{c}
\operatorname{Re}\beta_{\zeta}\\
\operatorname{Im}\beta_{\zeta}
\end{array}\right]\sim\mathcal{N}\left(\left[\begin{array}{c}
\operatorname{Re}\psi_{\zeta}\\
\operatorname{Im}\psi_{\zeta}
\end{array}\right]\mathbf{,}\left[\begin{array}{cc}
\sigma_{1}^{2} & 0\\
0 & \sigma_{2}^{2}
\end{array}\right]\right)
\end{equation}
where $\sigma_{1}^{2}\left(u\right)=\frac{1}{2N}\left[1+\psi_{\zeta}\left(2u\right)\right]-\frac{1}{N}\psi_{\zeta}^{2}\left(u\right)$
and $\sigma_{2}^{2}\left(u\right)=\frac{1}{2N}\left[1-\psi_{\zeta}\left(2u\right)\right]$.
\end{lemma}

\begin{proof} See \cite{Lo64}. \end{proof}

The joint distribution of $\operatorname{Re}\beta_{\xi}\left(u\right)$
and $\operatorname*{Im}\beta_{\xi}\left(u\right)$ can be obtained
similarly.

For a given $i\in\left\{ 2,\ldots,G\right\} $, using Lemma \ref{Lemma beampattern}
and the assumption that the mean patterns of both the transmitter
and receiver arrays satisfy (\ref{eq: eta zero}), i.e., $\psi_{\xi}\left(u_{1,i}\right)=\psi_{\xi}\left(2u_{1,i}\right)=\psi_{\zeta}\left(u_{1,i}\right)=\psi_{\zeta}\left(2u_{1,i}\right)=0$,
we have that, for both transmitter and receiver arrays, the array
pattern evaluated at any grid point is being drawn from an asymptotically
complex normal distribution with variance defined by the number of
transmit and receive elements, i.e., $\beta_{\xi}\left(u_{1,i}\right)\sim\mathcal{CN}\left(0,\frac{1}{M}\right)$
and $\beta_{\zeta}\left(u_{1,i}\right)\sim\mathcal{CN}\left(0,\frac{1}{N}\right)$.
It follows that the random variable $q=\frac{1}{N}\left\vert \mathbf{b}_{1}^{H}\mathbf{b}_{i}\right\vert $
can be approximated as belonging to Rayleigh distribution, i.e., $p\left(q\right)=\left(q/\sigma^{2}\right)\exp\left(-q^{2}/2\sigma^{2}\right)$,
where $\sigma^{2}=1/\left(2N\right)$, and similarly the random variable
$\frac{1}{N}\left\vert \mathbf{c}_{1}^{H}\mathbf{c}_{i}\right\vert $
is governed by a Rayleigh distribution with variance $\sigma^{2}=1/\left(2M\right)$.

If $\xi$ and $\zeta$ are independent (part 1), the two random variables
$\frac{1}{N}\left\vert \mathbf{b}_{1}^{H}\mathbf{b}_{i}\right\vert $
and $\frac{1}{N}\left\vert \mathbf{c}_{1}^{H}\mathbf{c}_{i}\right\vert $
are independent. Using (\ref{eq: aa bbcc}), we have that the distribution
of $\frac{1}{MN}\left\vert \mathbf{a}_{1}^{H}\mathbf{a}_{i}\right\vert $
is the product of two independent Rayleigh distributed variables.
The cumulative density function of such a variable is given in \cite{Simon02}.
It follows that the ccdf of $\frac{1}{MN}\left\vert \mathbf{a}_{1}^{H}\mathbf{a}_{i}\right\vert $
satisfies
\begin{equation}
\Pr\left(\frac{1}{MN}\left\vert \mathbf{a}_{1}^{H}\mathbf{a}_{i}\right\vert >q\right)<x\cdot K_{1}\left(x\right),
\end{equation}
where $x\triangleq2\sqrt{MN}q$.\\
If $\xi_{n}=\zeta_{n}$ for all $n$ (part 2), by using (\ref{eq: aa bbcc}),
we have that $\frac{1}{N^{2}}\left\vert \mathbf{a}_{1}^{H}\mathbf{a}_{i}\right\vert =\left(\frac{1}{N}\left\vert \mathbf{b}_{1}^{H}\mathbf{b}_{i}\right\vert \right)^{2}$.
Since the random variable $\frac{1}{N}\left\vert \mathbf{b}_{1}^{H}\mathbf{b}_{i}\right\vert $
has a Rayleigh distribution, $\frac{1}{N^{2}}\left\vert \mathbf{a}_{1}^{H}\mathbf{a}_{i}\right\vert $
is distributed as the square of a Rayleigh distribution, which has
cdf $1-\exp\left(-Nq\right)$.\ As such, its ccdf satisfies
\begin{equation}
\Pr\left(\frac{1}{N^{2}}\left\vert \mathbf{a}_{1}^{H}\mathbf{a}_{i}\right\vert >q\right)<\exp\left(-Nq\right).
\end{equation}
\\
Part 3 follows because, from (\ref{eq: aa bbcc}), the phase of $\mathbf{a}_{1}^{H}\mathbf{a}_{i}$
is the sum of the phases of $\mathbf{b}_{1}^{H}\mathbf{b}_{i}$ and
$\mathbf{c}_{1}^{H}\mathbf{c}_{i}$. In the case of transceivers the
phase of $\mathbf{a}_{1}^{H}\mathbf{a}_{i}$ is evidently uniform
since it is the same phase of $\mathbf{b}_{1}^{H}\mathbf{b}_{i}$.
In the case of $M$ transmitter and $N$ receivers, since both $\mathbf{b}_{1}^{H}\mathbf{b}_{i}$
and $\mathbf{c}_{1}^{H}\mathbf{c}_{i}$ are two independent circular
symmetric complex normal variables, the sum of the phases is itself
uniformly distributed over $\left[0,2\pi\right)$.

\subsection{Proof of Corollary \ref{Corollary Coherence}\label{App corollary coherence}}

We take the conservative assumption of independence between the $G-1$
random variables $\left\vert \frac{1}{MN}\mathbf{a}_{1}^{H}\mathbf{a}_{i}\right\vert $,
for $i=2,\ldots,G$. If $\xi$ and $\zeta$ are independent (part
1), from (\ref{eq: cdf mu}), the ccdf of the maximum among $G-1$
such variables (which gives the coherence), is upper bounded by
\begin{equation}
\Pr\left(\max_{i>1}\left\vert \frac{1}{MN}\mathbf{a}_{1}^{H}\mathbf{a}_{i}\right\vert >q\right)<1-\left[1-x\cdot K_{1}\left(x\right)\right]^{G-1},
\end{equation}
where $x\triangleq2\sqrt{MN}q$.\\
If $\xi_{n}=\zeta_{n}$ for all $n$ (part 2), by using (\ref{eq: cdf mu2}),
the ccdf of the maximum among $G-1$ such variables, is upper bounded
by
\begin{equation}
\Pr\left(\max_{i>1}\left\vert \frac{1}{N^{2}}\mathbf{a}_{1}^{H}\mathbf{a}_{i}\right\vert >q\right)<1-\left[1-\exp\left(Nq\right)\right]^{G-1}.
\end{equation}
This concludes the proof.

\subsection{Proof of Theorem \ref{Theorem MN coherence}\label{App MN coherence}}

The theorem follows by combining the claims of Theorem 2.7 in \cite{rauhut}
and Corollary \ref{Corollary Coherence}. Theorem 2.7 in \cite{rauhut}
provides stable recovery guarantees for any $K$-sparse signal if
the measurement matrix $\mathbf{A}$ has RIP $\delta_{2K}<2/\left(3+\sqrt{7/4}\right)\triangleq\alpha$.
The goal is therefore to bound the RIP $\delta_{2K}$ of the spatial
compressive sensing measurement matrix $\mathbf{A}$ with probability
higher than $1-\epsilon$. In other words, we want to find how many
measurements $MN$ we need to satisfy $\Pr\left(\delta_{2K}\leq\alpha\right)\geq1-\epsilon$.
By using $\delta_{2K}\leq\left(2K-1\right)\mu$ \cite{rauhut}, we
have that $\delta_{2K}\leq\alpha$ if $\left(2K-1\right)\mu\leq\alpha$.
Moreover, the condition $\Pr\left(\left(2K-1\right)\mu\leq\alpha\right)\geq1-\epsilon$
is equivalent to $\Pr\left(\mu>\alpha/\left(2K-1\right)\right)<\epsilon$.
Therefore by invoking (\ref{eq: upper bound mu}) in Corollary \ref{Corollary Coherence},
we can write
\begin{equation}
\Pr\left(\delta_{2K}>\alpha\right)<1-\left[1-x\cdot K_{1}\left(x\right)\right]^{G-1}
\end{equation}
where, by combining $x\triangleq2\sqrt{MN}q$ and $q=\alpha/\left(2K-1\right)$,
we have $x=2\sqrt{MN}\alpha/\left(2K-1\right)$. We thus look for
the value $MN$ that makes the right hand-side equal to $\epsilon$.

We first approximate the modified Bessel function of the second kind
$K_{1}\left(x\right)$ for a large absolute value and small phase
of the argument (in our setting, the argument $x$ is real) \cite{Abramoviz}:
$K_{1}\left(x\right)\approx\sqrt{\frac{\pi}{2x}}\exp\left(-x\right)$.
We thus would like to enforce $1-\left[1-\sqrt{\frac{\pi x}{2}}\exp\left(-x\right)\right]^{G-1}=\epsilon$.
Defining $t\triangleq\sqrt{\frac{\pi x}{2}}\exp\left(-x\right)$ and
linearizing the function $\left(1-t\right)^{G-1}$ around $t=0$,
we obtain $G\sqrt{\frac{\pi x}{2}}\exp\left(-x\right)=\epsilon$,
where, for simplicity, we used $G$ in place of $G-1$.\ This equation
can be rewritten in the form $-2x\exp\left(-2x\right)=-\gamma^{-2}$,
where $\gamma\triangleq\frac{\sqrt{\pi}G}{2\epsilon}$. The inverse
function of such equation is called the Lambert $W$ function \cite{LambertW}.
For real arguments, it is not injective, therefore it is divided in
two branches: $x>1/2$ or $x\leq1/2$. Since in our setup\ $x\leq1/2$,
the lower branch, denoted $W_{-1}$, is considered, and our solution
satisfies $-2x=W_{-1}\left(-\gamma^{-2}\right)$. By using the asymptotic
expansion $W_{-1}\left(-\gamma^{-2}\right)\approx-2\ln\gamma-\ln\left(2\ln\gamma\right)$
and solving for $MN$ we obtain (\ref{eq: MN coherence}). The claim
of the theorem\ follows from\ Theorem 2.7 in \cite{rauhut}. Finally,
since in this work we consider $K$-sparse signals, in the error term
(\ref{eq: error2}), we discarded the term for nearly-sparse signals
present in \cite{rauhut}.

\subsection{Proof of Theorem \ref{Lemma Isotropy}\label{App isotropy}}

Because the variables $\zeta_{1:N}$ are identically distributed,
and so are $\xi_{1:M}$, the average $\mathbb{E}\left[\mathbf{A}^{H}\left(t,:\right)\mathbf{A}\left(t,:\right)\right]$
does not depend on the index $t=N\left(m-1\right)+n$, where the last
relation follows from the definition of $\mathbf{a}_{g}$. Therefore,
we have $\mathbb{E}\left[\mathbf{A}^{H}\left(t,:\right)\mathbf{A}\left(t,:\right)\right]=\frac{1}{MN}\sum_{t=1}^{MN}\mathbb{E}\left[\mathbf{A}^{H}\left(t,:\right)\mathbf{A}\left(t,:\right)\right]=\frac{1}{MN}\mathbb{E}\left[\mathbf{Q}\right]$.
Thanks to Lemma \ref{Lemma Toeplitz}, we can focus only on the first
row of $\frac{1}{MN}\mathbb{E}\left[\mathbf{Q}\right]$. Using (\ref{eq: beta}),
the elements of the first row of such matrix are $\eta\left(u_{1,i}\right)$
for $i=1,\ldots,G$. From (\ref{eqn:mean_pat_app}) in Proposition
\ref{Lemma beampattern}, we know that $\eta\left(u_{1,i}\right)=\psi_{z}\left(u_{1,i}\right)$.
Thus, requiring (\ref{eq: dirac}), i.e., $\psi_{z}\left(u_{1,i}\right)=0$
for $i=2,\ldots,G$, together with the fact that $\exp\left(jzu_{1,1}\right)=1$
(because\ $u_{1,1}\triangleq\pi Z\left(\phi_{1}-\phi_{1}\right)=0$),
gives the ``if\textquotedblright{}\ direction of the claim.

The ``only if\textquotedblright{}\ direction follows by noticing
that when (\ref{eq: dirac}) is not satisfied there will be at least
one $i$ such that $\eta\left(u_{1,i}\right)\neq0$. Therefore, the
matrix $\mathbf{A}$ does not satisfy the isotropy property, showing
that (\ref{eq: dirac}) is also a necessary condition.


\begin{IEEEbiography}
[{\includegraphics[width=1in,height=1.25in,clip,keepaspectratio]{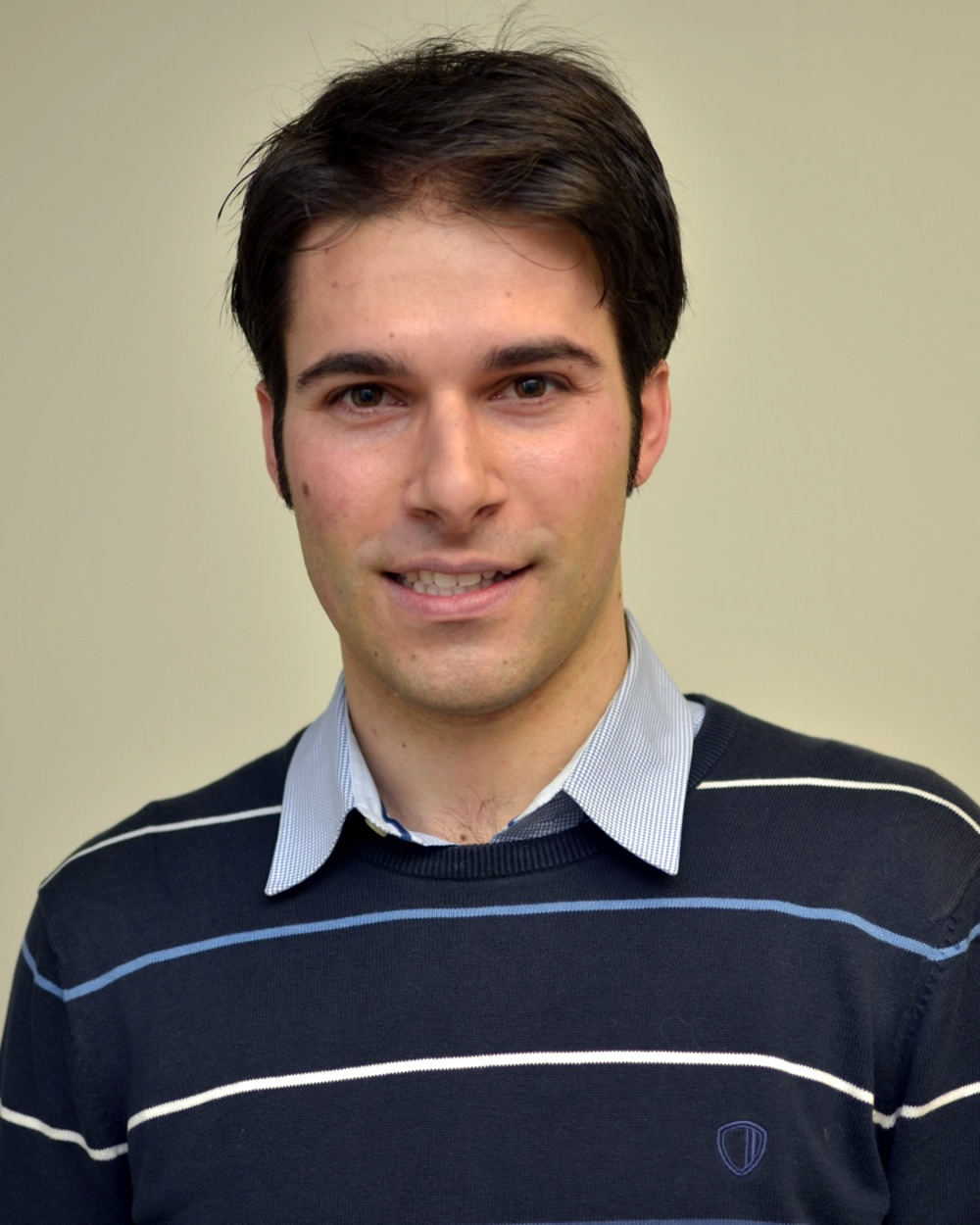}}]{Marco
Rossi} (S\textquoteright{}10) received the B.Sc. and M.Sc. degrees
in information engineering from Politecnico di Milano, Milan, Italy,
in 2004 and 2007, respectively. From November 2006 to November 2007,
as part of his M.Sc. thesis, he worked in the Radio System Technology
Research Group at Nokia-Siemens Network in Milan, Italy.

He is currently a Ph.D. candidate in electrical and computer engineering
at the Center for Wireless Communications and Signal Processing Research
(CWCSPR), New Jersey Institute of Technology (NJIT), Newark, NJ. His
research interests are in optimization methods, with an emphasis on
non-convex and discrete problems, and their applications to signal
processing and wireless communications. Marco Rossi was a recipient
of the Ross Fellowship Scholarship 2010\textendash{}12.
\end{IEEEbiography}
\vspace{.6in}
\begin{IEEEbiography}
[{\includegraphics[width=1in,height=1.25in,clip,keepaspectratio]{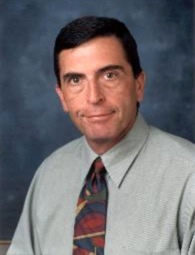}}]{Alexander
M. Haimovich} (S\textquoteright{}82-M\textquoteright{}87-SM\textquoteright{}97-F\textquoteright{}12)
received the B.Sc. degree in electrical engineering from the Technion-Israel
Institute of Technology, Haifa, Israel, in 1977, the M.Sc. degree
in electrical engineering from Drexel University, in 1983, and the
Ph.D. degree in systems from the University of Pennsylvania, Philadelphia,
in 1989. From 1983 to 1990 he was a design engineer and staff consultant
at AEL Industries. He served as Chief Scientist of JJM Systems from
1990 until 1992. He is the Ying Wu endowed Chair and a Professor of
Electrical and Computer Engineering at the New Jersey Institute of
Technology (NJIT), Newark, NJ, where he has been on the faculty since
1992. His research interests include MIMO radar, source localization,
communication systems, and wireless networks.
\end{IEEEbiography}

\vfill%\vspace{5in}% if you will not have a photo at all:

\begin{IEEEbiography}
[{\includegraphics[width=1in,height=1.25in,clip,keepaspectratio]{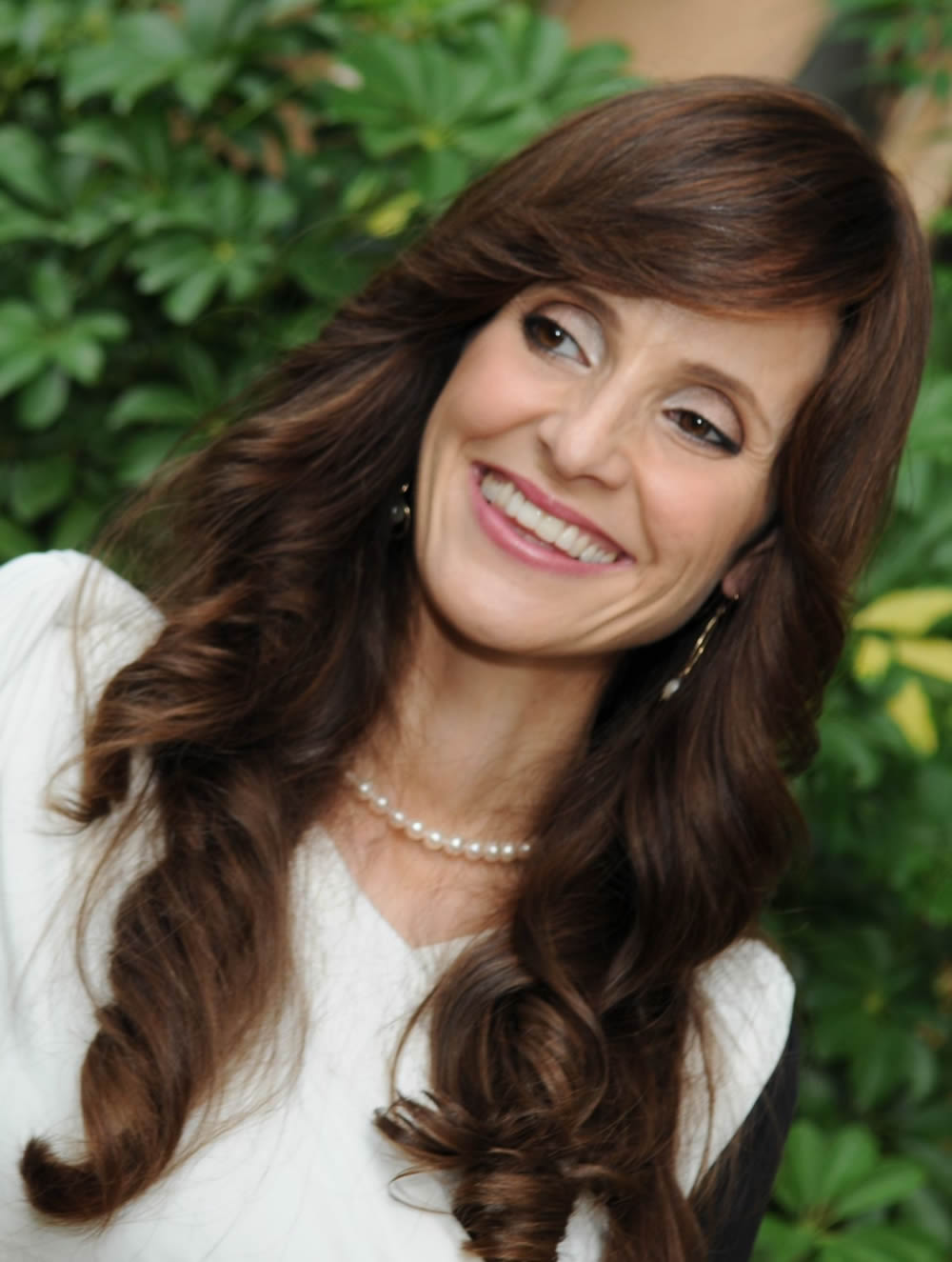}}]{Yonina
C. Eldar} (S\textquoteright{}98-M\textquoteright{}02-SM\textquoteright{}07-F\textquoteright{}12)
received the B.Sc. degree in physics and the B.Sc. degree in electrical
engineering both from Tel-Aviv University (TAU), Tel-Aviv, Israel,
in 1995 and 1996, respectively, and the Ph.D. degree in electrical
engineering and computer science from the Massachusetts Institute
of Technology (MIT), Cambridge, in 2002.

From January 2002 to July 2002, she was a Postdoctoral Fellow at the
Digital Signal Processing Group at MIT. She is currently a Professor
in the Department of Electrical Engineering at the Technion\textendash{}Israel
Institute of Technology, Haifa and holds the The Edwards Chair in
Engineering. She is also a Research Affiliate with the Research Laboratory
of Electronics at MIT and a Visiting Professor at Stanford University,
Stanford, CA. Her research interests are in the broad areas of statistical
signal processing, sampling theory and compressed sensing, optimization
methods, and their applications to biology and optics.

Dr. Eldar was in the program for outstanding students at TAU from
1992 to 1996. In 1998, she held the Rosenblith Fellowship for study
in electrical engineering at MIT, and in 2000, she held an IBM Research
Fellowship. From 2002 to 2005, she was a Horev Fellow of the Leaders
in Science and Technology program at the Technion and an Alon Fellow.
In 2004, she was awarded the Wolf Foundation Krill Prize for Excellence
in Scientific Research, in 2005 the Andre and Bella Meyer Lectureship,
in 2007 the Henry Taub Prize for Excellence in Research, in 2008 the
Hershel Rich Innovation Award, the Award for Women with Distinguished
Contributions, the Muriel \& David Jacknow Award for Excellence in
Teaching, and the Technion Outstanding Lecture Award, in 2009 the
Technion\textquoteright{}s Award for Excellence in Teaching, in 2010
the Michael Bruno Memorial Award from the Rothschild Foundation, and
in 2011 the Weizmann Prize for Exact Sciences. In 2012 she was elected
to the Youg Israel Academy of Science and to the Israel Committee
for Higher Education, and elected an IEEE Fellow. In 2013 she received
the Technion's Award for Excellence in Teaching, and the Hershel Rich
Innovation Award. She received several best paper awards together
with her research students and colleagues. She is a Signal Processing
Society Distinguished Lecturer, and Editor in Chief of Foundations
and Trends in Signal Processing. In the past, she was a member of
the IEEE Signal Processing Theory and Methods and Bio Imaging Signal
Processing technical committees, and served as an associate editor
for the IEEE Transactions On Signal Processing, the EURASIP Journal
of Signal Processing, the SIAM Journal on Matrix Analysis and Applications,
and the SIAM Journal on Imaging Sciences.
\end{IEEEbiography}

%\newpage


%\vfill


\end{document}